\documentclass[11pt]{article}

\usepackage[letterpaper, margin=1in]{geometry}
\usepackage{amsfonts,amssymb,amsmath,amsthm}
\usepackage[utf8]{inputenc}
\usepackage[T1]{fontenc}
\usepackage{newtxtext}
\usepackage[hidelinks]{hyperref}
\usepackage[shortlabels]{enumitem}
\usepackage{color}
\usepackage{tikz}
\usetikzlibrary{arrows,matrix,positioning}
\usepackage{tikz-cd}
\usepackage{graphicx,  fullpage, url}
\usepackage{verbatim}
\usepackage{mathrsfs}
\usepackage{mathtools}
\usepackage{array}
\usepackage{thm-restate}

\newcommand{\TODO}[1]{\textcolor{red}{\textbf{TODO: #1}} }

\theoremstyle{plain}
\newtheorem{thm}{Theorem}[section]
\newtheorem{lem}[thm]{Lemma}
\newtheorem{defi}[thm]{Definition}
\newtheorem*{defi*}{Definition}
\newtheorem{cor}[thm]{Corollary}
\newtheorem{claim}[thm]{Claim}
\theoremstyle{remark}
\newtheorem{rem}{Remark}
\newtheorem*{exmp*}{Example}

\newcommand{\N}{\mathbb{N}}
\newcommand{\Z}{\mathbb{Z}}

\newcommand{\m}{\mathfrak{m}}
\renewcommand{\O}{\mathcal{O}}

\newcommand{\F}{\mathbb{F}}
\newcommand{\btt}{\mathrm{BTT}}
\newcommand{\rs}{\mathrm{RS}}
\newcommand{\frs}{\mathrm{FRS}}
\newcommand{\ag}{\mathrm{AG}}
\newcommand{\dist}{\mathrm{dist}}
\newcommand{\rank}{\mathrm{rank}}
\newcommand{\poly}{\mathrm{poly}}
\renewcommand{\div}{\mathrm{div}}
\newcommand{\gal}{\mathrm{Gal}}
\newcommand{\Div}{\mathrm{Div}}
\newcommand{\supp}{\mathrm{Supp}}

\newcommand{\spn}{\ensuremath{\operatorname{span}}}

\newcommand{\img}{\ensuremath{\operatorname{Image}}}
\newcommand{\eps}{\epsilon}

\title{Efficient List-Decoding \\ with Constant Alphabet and List Sizes}

\author{
Zeyu Guo {\hskip 5em\relax}
Noga Ron-Zewi\thanks{Research supported in part by  ISF grant 735/20.} \\ 
Department of Computer Science, University of Haifa \\ \texttt{zguotcs@gmail.com}\qquad  \texttt{noga@cs.haifa.ac.il}
}

\date{}

\begin{document}
 
\maketitle

\begin{abstract}

We present an explicit and efficient algebraic construction of capacity-achieving list decodable codes with \emph{both} constant alphabet and constant list sizes. 
More specifically, for any $R \in (0,1)$ and $\epsilon >0$, we give an algebraic construction of
an infinite family of error-correcting codes of rate $R$, over an alphabet of size $(1/\epsilon)^{O(1/\epsilon^2)}$, that can be list decoded from  a $(1-R-\epsilon)$-fraction of errors with list size at most $\exp(\poly(1/\eps))$. Moreover, the codes can be encoded in time $\poly(1/\eps,n)$, the output list is contained in a linear subspace of dimension at most 
$\poly(1/\eps)$, and a basis for this subspace can be found in time $\poly(1 /\epsilon, n)$. Thus, both encoding and list decoding can be performed in \emph{fully polynomial-time} $\poly(1/\eps,n)$, except for pruning the subspace and outputting the final list which takes time $\exp(\poly(1/\eps)) \cdot \poly(n)$. In contrast, prior explicit and efficient constructions of capacity-achieving list decodable codes either required a much higher complexity in terms of $1/\eps$ (and were additionally much less structured), or had super-constant alphabet or list sizes.

Our codes are quite natural and structured. Specifically, we use algebraic-geometric 
($\ag$) codes with evaluation points restricted to a subfield, and with the message space restricted to a (carefully chosen) linear subspace.
Our main observation is that the output list of AG codes with subfield evaluation points
 is contained in an affine shift of the image of a \emph{block-triangular-Toeplitz $(\btt)$ matrix}, and that the list size can potentially be reduced to a constant by restricting the message space to a \emph{$\btt$ evasive subspace}, which is a large subspace that intersects the image of any $\btt$ matrix in a constant number of points. We further show how to explicitly construct such $\btt$ evasive subspaces, based on the explicit subspace designs of Guruswami and Kopparty (\emph{Combinatorica}, 2016), and composition.

\end{abstract}

\section{Introduction}

 An \emph{error-correcting code} is a map
$C: \Sigma^k \to \Sigma^n$, which encodes a $k$-symbol
\emph{message} over an \emph{alphabet} $\Sigma$ into an $n$-symbol \emph{codeword} over $\Sigma$. 
One main parameter of interest of an error-correcting code is the \emph{rate} $R = k/ n$, which measures
the amount of redundancy in the encoding.  
Naturally, it is desirable that the rate $R$ is as large as possible to minimize the overhead in encoding. Another important parameter is the
 \emph{ (relative) distance} $\delta$, defined as the smallest  \em (relative) Hamming distance\footnote{The \emph{(relative) Hamming distance} $\dist(z,w)$ between a pair of strings $z,w \in \Sigma^n$ is the fraction of coordinates on which $z$ and $w$ differ.}  \em
 $\dist( C(x), C(y))$ between the encodings of any pair of distinct messages $x, y \in \Sigma^k$. 
  The importance of the distance parameter
arises from the following observation:
if we are given $w \in \Sigma^n$ such that $\dist(w, C(x)) < \frac \delta 2$ for
some message $x \in \Sigma^k$, then this $x$ is \emph{uniquely determined}.  
Thus, a large distance allows one to \emph{unambiguously} retrieve the original message in the presence of some error or corruption. 
Other desirable properties of an error-correcting code are that its alphabet size would be small (ideally, a \emph{constant}, independent of the codeword length), and that it admits efficient ($\poly(n)$-time) encoding and decoding algorithms.

Clearly, there is a qualitative trade-off between the above parameters: the largest the distance $\delta$ is, the smallest the rate $R$ must be. Quantitatively, the \em Singleton bound \em states that any code must satisfy that $\delta \leq 1-R$. This bound is precisely matched by the classical family of \em Reed-Solomon $(\rs)$ codes \em \cite{RS60}. Given a finite field $\F_q$, and $n$ distinct elements $\alpha_1, \alpha_2, \ldots, \alpha_n \in \F_q$, the \em Reed-Solomon code $\mathsf{RS}_q(n,k)$ with evaluation points $\alpha_1, \ldots, \alpha_n$ \em maps 
a message $(f_0,f_1, \ldots f_{k-1}) \in \F_q^k$, viewed as the coefficients of a polynomial $f = \sum_{i=0}^{k-1} f_i X^i \in \F_q[X]_{<k}$, to the evaluation table
$(f(\alpha_1), \ldots, f(\alpha_n)) \in \F_q^n$. 

    A disadvantage of $\rs$ codes is that by definition, their alphabet size $q$ must be at least the  codeword length $n$. To match the Singelton bound over a \emph{constant-size} alphabet, independent of the codeword length $n$, one can resort to \em algebraic-geometric $(\ag)$ codes \em that achieve a distance of $\delta=1-R-\epsilon$ over a constant-size alphabet (depending on $\epsilon$) \cite{Sti09}. Moreover, both $\rs$ and $\ag$ codes can be efficiently encoded and decoded up to half their minimum distance \cite{Peterson60,BW,JLJHH89}.

\paragraph{List decoding.} In \em list decoding, \em the fraction of errors $\alpha$ is large enough so that unique recovery of the message $x$ is impossible (that is, $\alpha > \frac \delta 2$). Instead, the goal is, given a received word $w$, to return a short list $\mathcal{L}$ with the guarantee that $x \in \mathcal{L}$ for any message $x$ with $\dist(w,C(x))\leq \alpha$.  Besides being a fundamental concept in coding theory, list decoding has found diverse applications in theoretical computer science, for example in cryptography \cite{GL89}, learning theory \cite{KM93}, average-to-worst-case reductions \cite{CPS99, GRS00}, hardness amplification \cite{BFNW93, STV01, Tre03}, and pseudo-randomness \cite{TZ04, GUV09, DKSS13, TU12,GRX18}.

The \em list-decoding capacity theorem \em  states that the maximal fraction of errors for which list decoding with non-trivial list sizes is possible is $\alpha \leq 1-R$. Moreover, it is not hard to show that a random code of rate $R$ and alphabet-size $\exp(1/\epsilon)$ is with high probability list decodable from a $(1-R-\epsilon)$-fraction of errors with list size as small as $O(1/\epsilon)$.
So in principle, by allowing a small (constant-size) list, one can correct twice as many errors than in the unique decoding setting! However, matching these bounds with an explicit and efficient construction (ideally, encodable and list decodable in \emph{fully polynomial-time}\footnote{Note that in the list decoding setting, at least $\Omega(n/\eps)$ time is required to output the list. Moreover, the alphabet size must be at least $\exp(\Omega(1/\eps))$, and so the bit-length of the input is at least $\Omega(n/\eps)$. } $\poly(1/\eps, n)$)
turned out to be more challenging than in the unique decoding setting.

\paragraph{Capacity-achieving list decodable codes.}
The celebrated work of Guruswami and Sudan~\cite{Sudan97, GS-list-dec} showed that $\rs$ codes can be efficiently list decoded beyond half their minimum distance (up to the so-called \emph{Johnson bound}), which gave the first family of error-correcting codes that are efficiently list decodable beyond the unique decoding radius. Only a decade later,  the seminal work of Guruswami and Rudra \cite{GR08_folded_RS} showed that
\em folded Reed-Solomon $(\frs)$ codes \em -- a remarkably simple variant of $\rs$ codes -- can be efficiently list decoded up to list-decoding capacity. $\frs$ codes are obtained from $\rs$ codes (with the evaluation points ordered according to their power in the multiplicative group of the field) by dividing the codewords coordinates in the latter code 
 into consecutive blocks of length $m = \Theta(1/\epsilon^2)$, and then viewing each such block of coordinates as a single symbol over a larger alphabet. 

Once more, a disadvantage of $\frs$ codes is their large alphabet --
on the order of $n^{\Theta(1/\epsilon^2)}$ -- which is even larger  than that of the corresponding $\rs$ codes. 
Moreover, the list size obtained by the algorithm of Guruswami and Rudra was also a very large polynomial on the order of $n^{\Theta(1/\epsilon)}$, and this also dictated a similar running time for the list decoding algorithm. Starting with the breakthrough result of Guruswami and Rudra \cite{GR08_folded_RS}, there has been a long line of work attempting to construct explicit and efficient capacity-achieving list decodable codes with smaller alphabet and list sizes. Next, we briefly describe the main results of this line of work, and we refer the reader to Table \ref{tab:listdec} below for a summary of parameters.

\paragraph{Reducing list size.}
Towards reducing the list size, Guruswami and Wang \cite{GW13} devised a new ``linear-algebraic'' list decoding algorithm for $\frs$ codes, with the surprising property that the output list is contained in a low-dimensional subspace of \emph{constant} dimension $O(1/\epsilon)$. In the same work, Guruswami and Wang further
observed that, utilizing this property, one can potentially reduce the list size to a constant by restricting the message space of $\frs$ codes to a \emph{subspace evasive set}, which is a large set that intersects any constant dimensional subspace in a constant number of points. Guruswami and Wang showed that such objects exist probabilistically, and raised the question of searching for an explicit construction.  

The above program was subsequently carried out by Dvir and Lovet \cite{DL12}, who gave an algebraic construction of subspace evasive sets with the required properties. Combined with the linear-algebraic list decoding algorithm of \cite{GW13}, this resulted in a subcode of $\frs$ codes that can be efficiently encoded (in time $\poly(1/\eps, n)$), and efficiently list decoded up to capacity with constant list size $L=(1/\epsilon)^{O(1/\eps)}$ (in time $\poly(L,n)$). Lastly, Kopparty, Ron-Zewi, Saraf, and Wootters \cite{KRSW} have recently shown that in fact \emph{any} linear\footnote{A code $C:\Sigma^k \to \Sigma^n$ is \emph{linear} if $\Sigma= \F_q$ for some finite field $\F_q$, and the map $C: \F_q^k \to \F_q^n$ is linear.} code of constant distance $\delta$ that is list decodable from a $(\delta-\epsilon)$-fraction of errors with output list of constant dimension $d$ has constant list sizes (depending on $d$, $\delta$, and $\epsilon$). This shows that, perhaps surprisingly, $\frs$ codes themselves have constant list size (in fact the same list size of $L=(1/\epsilon)^{O(1/\eps)}$), without the need to pass to a subcode (list decoding can be performed \emph{probabilistically} in time $\poly(L,n)$).  

\paragraph{Reducing alphabet size.}
Similarly to the unique decoding setting, one can reduce the alphabet size to a constant (depending on $\eps$) by considering suitable versions of ``folded'' $\ag$ codes  \cite{Gur09, GX12, GX14, GX15}. 
However, in this setting, the dimension of the output list was too large to apply the above subspace evasive machinery and obtain small list sizes.
To overcome this, Guruswami and Xing \cite{GX13} came-up with an alternative approach for constructing capacity-achieving list decodable codes  that is based on restricting the evaluation points of ``plain'' (unfolded) versions of $\rs$ or $\ag$ codes to a subfield.\footnote{Guruswami and Xing first came-up with a similar approach in the folded setting \cite{GX12}, and only later observed in \cite{GX13} that it also applies to unfolded versions.
For simplicity, we only discuss the latter more basic approach.}

Specifically, Guruswami and Xing first observed that while $\rs$ codes are generally \emph{not} list decodable up to capacity with non-trivial list sizes \cite{BKR10}, for the special case of $\rs$ codes with evaluation points restricted to a subfield, it is possible to obtain slightly non-trivial list sizes, and furthermore, the lists satisfy a certain \emph{periodic} structure.
In more detail, consider the $\rs$ code $\mathsf{RS}_{q,m}(n,k)$, defined over a large extension field $\F_{q^m}$, with evaluation points coming from a small subfield $\F_q$, for $m = \Theta(1/\epsilon^2)$.   Guruswami and Xing showed that in this setting, there exists an $\F_q$-linear subspace $\hat V \subseteq \F_{q^m}$ of constant dimension $O(1/\epsilon)$ so that any message $(f_0,f_1, \ldots f_{k-1}) \in \F_{q^m}^k$ in the output list satisfies that once the first $i$ coefficients $f_0,f_1,\ldots,f_{i-1} \in \F_{q^m}$ are fixed, the next coefficient $f_i$ belongs to an affine shift of $\hat V$.

Note that, indeed, the above structure does not a priori guarantee a small list size. In fact, the only bound on the list  size that is implied by the above structure is $q^{O(k/\epsilon)} = q^{O(\epsilon m k)}$, which is only slightly smaller than the number of possible messages which is $q^{km}$. 
However, Guruswami and Xing noticed that, interestingly, the above periodic structure can lead to a list of constant dimension (which 
also leads in turn to constant list sizes using the machinery of \cite{DL12} or \cite{KRSW} described above) when the message space is restricted to a \emph{subspace design}, and once more, suggested to construct such objects explicitly.
In a follow-up work \cite{GK16}, Guruswami and Kopparty explicitly constructed such objects, and combined with the approach of Guruswami and Xing, this had the surprising consequence that a subcode of ``plain'' $\rs$ codes (with subfield evaluation points) is list decodable up to capacity, with constant list sizes (see Table \ref{tab:listdec} for exact parameters). 

Guruswami and Xing further observed that a similar periodic structure occurs in the $\ag$ code setting. However, over constant-size fields, it is impossible to construct subspace designs that 
 lead to constant list sizes.\footnote{In \cite{GX12}, Guruswami and Xing suggested alternatively using \emph{hierarchical subspace evasive sets} and showed that utilizing the above periodic structure, these could potentially lead to constant list sizes over constant-size alphabets. However, it is currently unknown how to explicitly construct such objects.} Nevertheless, Guruswami and Xing showed  
 that one can iteratively compose together subspace designs of exponentially increasing lengths to obtain extremely slowly growing list sizes (depending on $\log^*n$) over constant-size alphabets (see more discussion in Section \ref{subsec:tech} below). 
 This led in turn to capacity-achieving list decodable codes with constant alphabet size $(1/\eps)^{O(1/\eps^2)}$,  extremely slowly growing list size  $\exp(\poly(1/\eps)) \cdot \exp\exp\exp(\log^*n)$, and efficient encoding and list decoding algorithms (running in time $\poly(1/\eps, n)$ and $\poly(L,n)$, respectively).
 
Finally, we mention that in  \cite{KRSW}, a different approach was given for obtaining \emph{both} constant list and constant alphabet sizes, based on multi-level concatenation of $\frs$ codes, and expander-based amplification. However, the resulting code is arguably more complicated and less natural and structured than the aforementioned algebraic constructions, and moreover, has a much higher complexity in terms of $1/\eps$. Specifically, the list size was \emph{quadruply-exponential} in $\poly(1/\eps)$, which also dictated a similar running time for list decoding,\footnote{The reason for the large list size is that the construction roughly uses four levels of encodings, two of these via $\frs$ codes, and two other via random linear codes, and for both codes, the best-known list size is exponential in $1/\eps$. It may be possible to reduce the list size by replacing the random linear codes with other codes of smaller list size and succinct representation, e.g., the \emph{pseudo-linear codes} of  \cite{GI01}. However, the list size would still be at least doubly-exponential in $1/\eps$.}, and the encoding time was also pretty large $\exp(\poly(1/\eps)) \cdot \poly(n)$ due to the need to brute-force search for the inner codes.

\subsection{Our results}
A main question left open by the above line of work is whether one can come up with constructions of capacity-achieving list decodable codes with \emph{both} constant alphabet and constant list sizes, and admitting \emph{fully polynomial-time} $\poly(1/\eps,n)$ encoding and list-decoding. Our main result (almost) answers this question in the affirmative.

\begin{thm}\label{thm:main}
For any $R \in (0,1)$ and $\epsilon >0$, there is an infinite family of error-correcting codes of rate  at least  $R$ over an alphabet of size $(1/\epsilon)^{O(1/\epsilon^2)}$ that can be encoded in time $\poly(1/\eps,n)$, and can be list decoded from  a $(1-R-\epsilon)$-fraction of errors with list size at most $L=\exp(\poly(1/\epsilon))$ in time $\poly(L,n)$. Moreover, the codes, defined over an alphabet $\F_{q^m}$, are $\F_q$-linear, the output list is contained in an $\F_q$-linear subspace of dimension at most $\poly(1/\epsilon)$, and a basis for this subspace can be found in time $\poly(1 /\epsilon, n)$. 
\end{thm}

Note that our codes achieve list-decoding capacity with \emph{both} constant alphabet and constant list sizes, and both encoding and list decoding can be performed in \emph{fully polynomial-time} $\poly(1/\eps,n)$, except for pruning the subspace and outputting the final list which takes time $\exp(\poly(1/\eps)) \cdot \poly(n)$. Our codes are quite natural and structured, specifically, we use
$\ag$ codes with evaluation points restricted to a subfield, and with the message space restricted to a (carefully chosen) $\F_q$-linear subspace. It is our hope that this relatively natural and simple structure will prove useful in future applications.

A barrier to improving the general running time of list-decoding to $\poly(1/\eps,n)$
is the exponential dependency of the list size on $1/\eps$ since, at the very least, such amount of time is required to output the whole list. However, currently the smallest known list size for explicit capacity-achieving list decodable codes is $(1/\epsilon)^{O(1/\epsilon)}$, achieved by $\frs$ codes \cite{KRSW}. We leave it as an interesting open problem to search for explicit capacity-achieving list decodable codes (even over large super-constant alphabet) with optimal list size $O(1/\epsilon)$, or even $\poly(1/\epsilon)$. 
We further mention that the alphabet size we obtain, on the other hand, is not much worse than the lower bound of $\exp( \Omega(1/\epsilon))$, and is generally the smallest known alphabet size for explicit capacity-achieving list decodable codes \cite{GX12, GX13, GX14, GX15}.

Finally, we note that using the machinery of \cite{HRW, KRRSS} (specifically, taking a high-order tensor product of the codes given by Theorem \ref{thm:main}, combined with an expander-based amplification), it is possible to bring down the dependency on $n$ in the running time of both encoding and list-decoding to \emph{nearly-linear}, say $n^{1.01}$. However, similarly to the multi-level construction of \cite{KRSW} mentioned above, the resulting codes become 
more complicated and
less natural and structured, and also have a much higher complexity in terms of $1/\eps$.  
Obtaining a \emph{truly-linear}  dependency on $n$ in the running time of either encoding or list-decoding for capacity-achieving list decodable codes seems to require, as in the unique decoding setting, completely different non-algebraic techniques. 

Table \ref{tab:listdec} below summarizes the above discussion. In the next section, we shall give an overview of our techniques.
\begin{center}
\begin{table}
\begin{tabular}{|m{4.5cm}|c|c|m{3cm}|}
\hline 
\center{\textbf{Code}} &\textbf{Alphabet size $|\Sigma|$} & \textbf{List size $L$}& \textbf{Notes} \\ \hline  \hline
Random codes & $2^{O(1/\epsilon)}$&  $O(1/\eps)$ &   Non-constructive \\ \hline \hline
$\frs$ codes \cite{GR08_folded_RS, GW13} & $n^{O(1/\eps^2)}$&  $n^{O(1/\eps)}$ &   \\ \hline 
Previous codes  \cite{KRSW} & $n^{O(1/\eps^2)}$& $\left(\frac{1}{\eps}\right)^{O(1/\eps)}$ &   Randomized list-decoding \\ \hline
Previous codes + subspace evasive set \cite{DL12} & $n^{O(1/\eps^2)}$ & $\left(\frac{1}{\eps}\right)^{O(1/\eps)}$ &   \\ \hline 
Multi-level concatenation of previous codes + expander amplification \cite{KRSW} & $2^{\poly(1/\eps)}$& $2^{2^{2^{2^{\poly(1/\epsilon)}}}}$ &   
Encoding time $2^{\poly(1/\eps)} \cdot \poly(n)$
\\ \hline \hline
$\rs$ codes with subsfield evaluation points + subspace design \cite{GX13, GK16} &  $n^{O(1/\eps^2)}$ & $n^{O(1/\epsilon^3)}$ &  \\ \hline
Previous codes \cite{KRSW} &  $n^{O(1/\eps^2)}$ & $\left(\frac{1}{\eps}\right)^{O(1/\eps^4)}$ &   Randomized list-decoding\\ \hline
Previous codes + subspace evasive set \cite{DL12} &  $n^{O(1/\eps^2)}$ & $\left(\frac{1}{\eps}\right)^{O(1/\eps^3)}$ &   \\ \hline \hline
$\ag$ codes with subfield evaluation points + subspace design \cite{GX13, GK16} &$\left(\frac{1}{\eps}\right)^{O(1/\eps^2)}$ & $2^{\poly(1/\eps)} \cdot 2^{2^{2^{O(\log^*n)}}}$ &  \\ \hline
Tensor product of previous codes + expander amplification \cite{HRW,KRRSS}  & $2^{\poly(1/\eps)}$ & $2^{2^{2^{2^{^{\poly(1/\eps)}}}}}\cdot 2^{2^{2^{2^{O(\log^*n)}}}}$ &  Encoding time $2^{\poly(1/\eps)} \cdot n^{1.01}$, list-decoding time $\poly(L) \cdot n^{1.01}$   \\ 
\hline \hline
\textbf{This work:} $\ag$ codes with subfield evaluation points + $\btt$ evasive subspace & $\left(\frac{1}{\eps}\right)^{O(1/\eps^2)}$& $2^{\poly(1/\eps)}$ &   
Basis for subspace containing list can be found in time $\poly(1/\eps,n)$
\\ \hline
%
\end{tabular}

\caption{Capacity-achieving list decodable codes $C: \Sigma^k \to \Sigma^n$ of rate $R$ that are list decodable from a $(1-R-\epsilon)$-fraction of errors with list size $L$. 
All codes can be deterministically encoded in time $\poly(1/\eps,n)$ and deterministically list decoded in time $\poly(L,n)$ unless otherwise noted.}
 \label{tab:listdec}
\end{table}
\end{center}

\section{Techniques}\label{subsec:tech}

The starting point for our construction is the aforementioned work of Guruswami of Xing \cite{GX13}. As described above, in this work it was observed that the output list of $\rs$ or $\ag$ codes with subfield evaluation points satisfy a special \emph{periodic} structure. Namely, there exists an $\F_q$-linear subspace $\hat V \subseteq \F_{q^m}$ of constant dimension $r=O(1/\epsilon)$ so that any message $(f_0,f_1, \ldots f_{k-1}) \in \F_{q^m}^k$ in the output list satisfies that given the first $i$ coefficients $f_0,f_1,\ldots,f_{i-1} \in \F_{q^m}$, the next coefficient $f_i$ belongs to an affine shift of $\hat V$.
Moreover, it was shown that one can exploit this structure and reduce the output list size by restricting the message space to a \emph{subspace design}. 

An \emph{$(r,s)$-subspace design over $\F_{q^m}$ of cardinality $k$} is a collection of $k$ $\F_q$-linear subspaces $H_1, \ldots, H_k \subseteq \F_{q^m}$ so that $\sum_{i=1}^{k} \dim(\hat V \cap H_i)\leq s$ for any $\F_q$-linear subspace  $\hat V \subseteq \F_{q^m}$ of dimension at most $r$. It follows by definition that, assuming the above periodic structure, when restricting each coefficient $f_i$ to the subspace $H_i$, the resulting output list has dimension at most $\sum_{i=1}^{k} \dim(\hat V \cap H_i) \leq s$. It can be shown, using the probabilistic method, that there exists an $(r,s)$-subspace design $H_1, \ldots, H_k$ over $\F_{q^m}$ with $k  = q^{\Omega(\epsilon m)}$ and $s = O(r/\epsilon)$, where each subspace $H_i$ has co-dimension at most $\epsilon m$ in $\F_{q^m}$. In \cite{GK16}, Guruswami and Kopparty gave an explicit construction of 
a subspace design with similar parameters.  

\begin{restatable}[Explicit subspace design, \cite{GK16}, Theorem 6]{thm}{sdexp}\label{thm:subspace-design-explicit}
There exists an absolute constant $c >1$, so that 
for every $\epsilon >0$, positive integers $k,m,r$ with $r < \frac{\epsilon m} {4}$, and a prime power $q$ satisfying
$q^m \geq \max\left\{k^{c \cdot r /\epsilon}, \left(\frac {2r} {\eps}\right) ^{2 r /\epsilon}\right\}$,
 there exists an $(r,s)$-subspace design $H_1, \ldots, H_k$ over $\F_{q^m}$ for $s= \frac{2r^2} {\epsilon}$, where each $H_i$ has co-dimension at most $\epsilon m$ in $\F_{q^m}$.
  Moreover, bases for  $H_1, \ldots, H_k$ can be found in time $\poly(q,k,m)$.
\end{restatable}

Thus, by restricting the message coefficients in $\rs$ codes with subfield evaluation points to the subspace design given by the above theorem, one can reduce the dimension of the output list to $O(1/\epsilon^3)$ (and by \cite{KRSW}, this in fact implies that the list is of constant size). 
Note however that  the above theorem could not be applied to $\ag$ codes, as it requires the number of subspaces $k$ to be smaller than $q^m$, whereas for $\ag$ codes the number of message coordinates $k$ (which grows to infinity) is much larger than $q^m$ (which is constant in the $\ag$ setting).  

To overcome this, Guruswami and Xing suggested the following iterative construction. Suppose that the message space has the periodic structure described above, and that $k \gg q^m$. Then Guruswami and Xing suggested to first divide the $k$ coordinates
into $\frac k {k_1}$ blocks of $k_1$ coordinates each, where $k_1 \approx q^m$, and restrict each such block separately to an identical $(r,s)$-subspace design over $\F_{q^m}$ of cardinality $k_1$
that is guaranteed by the above Theorem \ref{thm:subspace-design-explicit}. The main observation is that, when viewing each block of length $k_1$ as a single coordinate, the resulting subspace also has a periodic structure, however, with exponentially larger alphabet size $q^{mk_1}$. 

Thus, one can once more divide the resulting  $\frac k {k_1}$ coordinates into $\frac k {k_1 k_2}$ blocks of length $k_2$ each, where now $k_2 \approx q^{m k_1}$, and restrict to an identical subspace design on each block separately. 
Continuing this way, and noting that the alphabet size increases exponentially in each iteration, after $\approx \log^*k$ iterations, we arrive at alphabet size $k$, which is sufficiently large for restricting to a single subspace design. Since the dimension squares on each invocation of Theorem \ref{thm:subspace-design-explicit}, the final dimension is doubly-exponential in $\log^*k$, and the resulting output list size is triply-exponential $\log^*k$.


Our main observation that allows us to obtain \emph{both} constant alphabet and constant list sizes is that $\ag$ (or $\rs$) codes with subfield evaluation points satisfy yet an even more refined structure, namely, the output list is contained in an affine shift of the image of a  \emph{block-triangular-Toeplitz $(\btt)$ matrix}. We further observe that this structure can potentially lead to constant list sizes over a constant-size alphabet if the message space is restricted to an appropriate pseudo-random object that we call a \emph{$\btt$ evasive subspace}. In what follows we elaborate on these two ingredients. 

\subsection{Block-triangular-Toeplitz matrix}

We start by formally defining the notion of a block-triangular-Toeplitz matrix  (see Figure \ref{fig_btt} below for an illustration).
\begin{defi}[Block-triangular-Toeplitz ($\btt$) matrix]
A \emph{$(k,m,r)$-block-triangular-Toeplitz ($\btt$) matrix over $\F_q$} is a $(km) \times (kr)$ matrix $M$ over $\F_q$ so that $M=(M_{i,j})_{i,j\in [k]}$, as a $(k\times k)$-block matrix with $m \times r$ blocks $M_{i,j}$, satisfies the following conditions:
\begin{enumerate}
\item $M$ is block-lower-triangular, i.e., $M_{i,j}=0$ for $i,j\in [k]$ with $i<j$.
\item $M$ is block-Toeplitz, i.e.,  $M_{i,j}=M_{i',j'}$ for $i,j,i',j'\in [k]$ with $i-j=i'-j'$.
\item $M$ has maximal rank.  By the two conditions above, this is equivalent to the statement that $M_{1,1}$ has rank $\min\{r,m\}$. 
\end{enumerate} 
We say that  $M$ is $(k,m,r)$-\emph{periodic} if only blocks on the main diagonal are required to be identical, i.e., 
the second condition above is weakened to  $M_{i,i}=M_{i',i'}$ for all $i,i'\in [k]$.
\end{defi}

 \begin{figure}[htb]
\centering
\[
\begin{pmatrix}
M_{1} & 0 & 0 & \cdots & 0  \\
M_{2} & M_{1} & 0 & \cdots & 0 \\
M_{3} & M_{2} & M_{1} & \cdots & 0 \\
\vdots  & \vdots  & \vdots & \ddots & \vdots  \\
M_{k} & M_{k-1} & M_{k-2} & \cdots & M_{1} 
\end{pmatrix}
\]
\caption{A $(k, m, r)$-$\btt$ matrix, where each $M_i$ is an $m\times r$ matrix and $M_1$ has maximal rank.}
\label{fig_btt}
\end{figure}

We say that $V\subseteq \F_q^{k m}$  is a \emph{$(k,m,r)$-$\btt$ subspace}
if $V=\img(M)$ for some $(k,m,r)$-$\btt$ matrix $M$ (where $M$ is viewed as a linear map from $\F_q^{kr}$ to $\F_q^{km}$, and $\img(M)$ denotes its image). Similarly, we say that $V\subseteq \F_q^{k m}$  is a \emph{$(k,m,r)$-periodic subspace}
if $V=\img(M)$ for some $(k,m,r)$-periodic matrix $M$.
Note that in this terminology, the periodic structure described above corresponds to the special case of a $(k,m,r)$-periodic subspace, where $\hat V = \img(M_1)$.
Our first main observation is that the output list of $\ag$ (or $\rs$) codes with subfield evaluation points is in fact contained in an affine shift of
a $(k,m,r)$-$\btt$ subspace for $r = O(1/\epsilon)$ (under a suitable linear map).

\begin{restatable}[Output list contained in a $\btt$ subspace]{thm}{agbtt}\label{thm:intro_ag_btt}
There exists an absolute constant $c >1$ so that the following holds for any  $R \in (0,1)$,  $\epsilon >0$, $q \geq 1/\epsilon^c$ that is an even power of a prime, and $m \geq 1/\epsilon^2$. There is an infinite family of error-correcting codes $\{C_n\}_{n}$, where $C_n$ satisfies the following properties:
\begin{enumerate}
\item $C_n: \F_{q^m}^k \to \F_{q^m}^n$ is a linear code of rate at least $R$ that can be encoded in time $\poly(\log q, m,n)$.
\item There exists an injective $\F_q$-linear map $\phi: \F_{q^m}^k \to \F_{q^m}^{\hat k}$, where $\hat k \leq n$, so that $C_n$ can be list decoded from a $(1-R-\epsilon)$-fraction of errors, pinning down the images of the candidate messages under $\phi$ (viewed as length $\hat k m$ vectors over $\F_q$) to an affine shift of a $(\hat k,m,\epsilon m)$-$\btt$ subspace $V$ over $\F_q$. Moreover, the map $\phi$, a basis for $V$, and the affine shift can be computed in time $\poly(\log q, m,n)$.
  \end{enumerate}
\end{restatable}

We prove the above theorem in Section \ref{sec:AG} using $\ag$ codes with subfield evaluation points. As a warm-up, we first prove, in Section \ref{sec:RS},  that this theorem holds in the more basic setting of $\rs$ codes with subfield evaluation points. In the $\rs$ setting,
the linear-algebraic approach of \cite{GX13} gives a functional equation of the form
$$A_0(X)+A_1 (X)f(X)+ A_2(X) f^{\sigma}(X)+\dots+A_s (X) f^{\sigma^{s-1}}(X)=0$$
that any low-degree polynomial $f$ that has large agreement with a received word  $\mathbf{y}$ must satisfy, where $\sigma$ denotes
 the \emph{Frobenius automorphism} mapping $ f= \sum_{j=0}^{k-1}f_jX^j$ to $ f= \sum_{j=0}^{k-1}f^q_jX^j$, and the coefficients of $A_0, \ldots, A_s$ depend on the received word $\mathbf{y}$. 
 
We observe that the above functional equation  quite naturally gives a $(k,m,(1-\epsilon)m)$-$\btt$ matrix, so that the solution set is contained in an affine shift of the \emph{kernel} of this matrix.\footnote{The fact that the solution set is contained in an affine shift of the kernel of a $\btt$ matrix was also commented (but not exploited) in \cite[Definition 3]{GX20}. Here we provide a formal proof for this fact, show that the solution set can be equivalently defined as the image of a $\btt$ matrix, and show that this property can be exploited to obtain improved list sizes.} We then show that the kernel of a $(k,m,r)$-$\btt$ matrix  is in fact a $(k,m,m-r)$-$\btt$ subspace (i.e., the \emph{image} of a $(k,m,m-r)$-$\btt$ matrix), which gives the claimed $(k,m,\epsilon m)$-$\btt$ subspace containing the list (in this setting, $\phi$ is just the identity map). We further show that a similar reasoning can be applied in the $\ag$ code setting.

\subsection{BTT evasive subspace}
We say that a subspace  $W\subseteq \F_q^{km}$ is a  \emph{$(k,m,r,s)$-$\btt$ evasive subspace}  if $\dim (V\cap W)\leq s$ for every $(k,m,r)$-$\btt$ subspace $V\subseteq \F_q^{k m}$. 
Similarly, we say that a subspace  $W\subseteq \F_q^{km}$ is a  \emph{$(k,m,r,s)$-periodic evasive subspace}  if $\dim (V\cap W)\leq s$ for every $(k,m,r)$-periodic subspace $V\subseteq \F_q^{k m}$. Note that any $(r,s)$-subspace design over $\F_{q^m}$ of cardinality $k$ is a $(k,m,r,s)$-periodic evasive subspace (see Corollary \ref{cor:periodic-explicit}).
We first observe, using the probabilistic method, that there exists a $(k,m,r,s)$-$\btt$ evasive subspace $W \subseteq \F_q^{km}$ of co-dimension at most $\epsilon km$ and $s = O\left(r/ \epsilon\right)$. Notably, the lemma holds for any field size $q$ and block length $m$!

\begin{restatable}{lem}{bttnonexp}\label{lem:btt_non_explicit}
For every $\epsilon >0$, positive integers $k,m,r$ with $r < \frac{\epsilon m} {2}$,  and a prime power $q$, there exists a 
$(k,m,r,s)$-$\btt$ evasive subspace $W \subseteq \F_q^{km}$ of co-dimension at most $\epsilon km$ for $s= \frac{2r} {\epsilon}$. 
\end{restatable}

Moreover, we are able to explicitly construct $\btt$ evasive subspaces with similar parameters over a field of  considerably smaller size than the one required in Theorem \ref{thm:subspace-design-explicit}.

\begin{restatable}[Explicit $\btt$ evasive subspace]{thm}{bttexp}\label{thm:btt_explicit}
There exists an absolute constant $c>1$, so that for every $\epsilon >0$, positive integers $k,m,r$ with $r < \frac{\epsilon m} {24}$,  and a prime power $q$ satisfying that $q \geq m^c$, there exists a 
$(k,m,r,s)$-$\btt$ evasive subspace $W \subseteq \F_q^{km}$ of co-dimension at most $\epsilon k m$ for $s=  \poly(r/\epsilon)$. Moreover, a basis for  $W$ can be found in time $\poly(q,k,m)$. 
\end{restatable}

We prove the above theorem in Section \ref{sec:btt}. To this end, we first observe that the iterative construction of \cite{GX13}, described above, implicitly gives the following \emph{composition lemma} for \emph{periodic} evasive subspaces. In what follows, for a subspace $W \subseteq \F_q^n$, and a positive integer $k$, let $W^k \subseteq \F_q^{k n}$ be the subspace containing all vectors $(w_1, \ldots, w_k) \in \F_q^{kn}$ where $w_i \in W$ for all $1\leq i \leq k$. For a pair of subspaces $W \subseteq \F_q^n$ and $W' \subseteq \F_q^{kn}$, we let $W \circ W':=W^{k} \cap W'$. 

\begin{lem}[Implicit in \cite{GX13}]
Suppose that $W$ is an ``inner'' $(k,m,r,s)$-periodic evasive subspace over $\F_q$, and $W'$ is an ``outer'' $(k',k m,s, s' )$-periodic evasive subspace over $\F_q$. Then $W \circ W'=W^{k} \cap W'$ is a $(k' k, m, r ,s' )$-periodic evasive subspace over $\F_q$.
\end{lem}

Roughly speaking, the above lemma gives a way to combine together an ``inner'' periodic evasive subspace $W$ with a short block length $m$ (but a relatively small number of blocks $k$) with an ``outer'' periodic evasive subspace $W'$ with a large number of blocks $k'$ (but a long block length $m'$), to obtain a periodic evasive subspace $W \circ W'$ of both short block length $m$ and large number of blocks $\approx k'$ (see Figure \ref{fig_composition} below for an illustration). Applying this lemma iteratively for $\log^*k$ times, using the explicit subspace design given by Theorem \ref{thm:subspace-design-explicit} (which is in particular a periodic evasive subspace), gives the main result of \cite{GX13} which reduces the list size of $\ag$ codes with subfield evaluation points to triply-exponential in $\log^*n$.

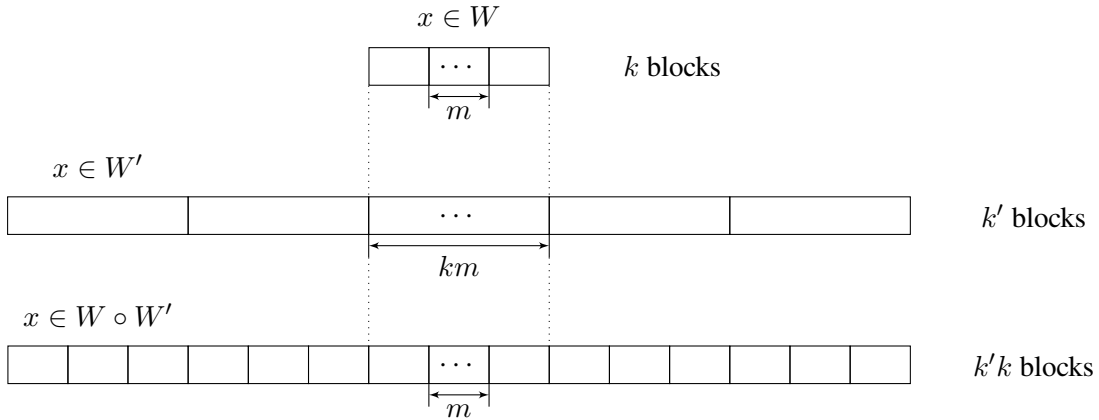
\begin{figure}[htb]
\centering
\begin{tikzpicture}[node distance=1.2cm,
    field/.style={draw,  minimum height=5mm, anchor=center, align=center},
    frame/.style={matrix of nodes, column sep=-\pgflinewidth, nodes={field}},
    ]
   \newcommand\lw{\pgflinewidth/2}

    \matrix (F1) [frame, minimum width=8mm]{
    {} & \dots & {}    \\};

    \matrix (F2) [frame, minimum width=24mm, below=of F1]{
    {} & {} &  \dots & {} & {} \\};
    
    \matrix (F3) [frame, minimum width=8mm, below=of F2]{
    {} & {} & {} & {} & {} & {} & {} &  \dots & {} & {} & {} & {} & {} & {} & {} \\};
    
    \node at ([yshift=4mm]F1-1-2.north) {$x\in W$};
    \node at ([yshift=4mm]F2-1-1.north) {$x\in W'$};
    \node at ([yshift=4mm]F3-1-2.north) {$x\in W\circ W'$};
    
    \node at ([xshift=15mm]F1.east) {$k$ blocks};
    \node at ([xshift=15mm]F2.east) {$k'$ blocks};
    \node at ([xshift=15mm]F3.east) {$k'k$ blocks};
 
    \draw[dotted] ([xshift=\lw]F1-1-1.south west) -- ([xshift=\lw]F2-1-3.north west);
    \draw[dotted] ([xshift=-\lw]F1-1-3.south east) -- ([xshift=-\lw]F2-1-3.north east); 
    \draw[dotted] ([xshift=\lw, yshift=-3mm]F2-1-3.south west) -- ([xshift=\lw]F3-1-7.north west);
    \draw[dotted] ([xshift=-\lw, yshift=-3mm]F2-1-3.south east) -- ([xshift=-\lw]F3-1-9.north east); 
    
    \draw[latex'-latex'] ([xshift=\lw, yshift=-4mm]F1-1-2.west) -- ([xshift=-\lw, yshift=-4mm]F1-1-2.east) node[below, midway] {$m$};
    \draw[latex'-latex'] ([xshift=\lw, yshift=-4mm]F2-1-3.west) -- ([xshift=-\lw, yshift=-4mm]F2-1-3.east) node[below, midway] {$km$};
    \draw[latex'-latex'] ([xshift=\lw, yshift=-4mm]F3-1-8.west) -- ([xshift=-\lw, yshift=-4mm]F3-1-8.east) node[below, midway] {$m$};
    \draw[] ([xshift=\lw]F1-1-2.south west) -- ([xshift=\lw, yshift=-3mm]F1-1-2.south west);
    \draw[] ([xshift=-\lw]F1-1-2.south east) -- ([xshift=-\lw, yshift=-3mm]F1-1-2.south east);
    \draw[] ([xshift=\lw]F2-1-3.south west) -- ([xshift=\lw, yshift=-3mm]F2-1-3.south west);
    \draw[] ([xshift=-\lw]F2-1-3.south east) -- ([xshift=-\lw, yshift=-3mm]F2-1-3.south east);
    \draw[] ([xshift=\lw]F3-1-8.south west) -- ([xshift=\lw, yshift=-3mm]F3-1-8.south west);
    \draw[] ([xshift=-\lw]F3-1-8.south east) -- ([xshift=-\lw, yshift=-3mm]F3-1-8.south east);
\end{tikzpicture}
\caption{Illustration of the first two parameters in composition.}
\label{fig_composition}
\end{figure}

We further observe that essentially the same composition lemma holds when replacing the inner periodic evasive subspace with a \emph{$\btt$ evasive subspace}, in which case the resulting composed subspace is a \emph{$\btt$} evasive subspace as well.

\begin{restatable}{lem}{composition}\label{lem:composition}
Suppose that $W$ is an ``inner'' $(k,m,r,s)$-$\btt$ evasive subspace over $\F_q$, and $W'$ is an ``outer'' $(k',k m,s, s' )$-periodic evasive subspace over $\F_q$. Then $W \circ W'=W^{k'} \cap W'$ is a $(k' k, m, r ,s' )$-$\btt$ evasive subspace over $\F_q$.
\end{restatable}

To prove Theorem \ref{thm:btt_explicit}, we first apply the above composition lemma with the inner subspace being the
non-explicit $(k_1,m,r,s_1)$-$\btt$ evasive subspace, given by Lemma \ref{lem:btt_non_explicit}, for $k_1 \approx \log \log k$ (which can be found efficiently via brute-force search in this setting of parameters), and the outer subspace being the explicit $(k_2,k_1m,s_1,s_2)$-periodic evasive subspace, given by Theorem \ref{thm:subspace-design-explicit}, for $k_2 = \log k\approx \exp(k_1)$. Then we apply the above composition lemma once more with the inner subspace being the resulting 
$(k_2 k_1, m,r,s_2)$-$\btt$ evasive subspace, and the outer subspace being yet another explicit $(k_3,k_1k_2m,s_2,s_3)$-periodic evasive subspace, given by Theorem \ref{thm:subspace-design-explicit}, for $k_3 \approx k\approx \exp(k_2) $. As we apply the composition step only twice, this results in a $(k,m,r,s_3)$-$\btt$ evasive subspace for $s_3 = \poly(r/\eps)$. A careful choice of parameters gives the explicit $\btt$ evasive subspace claimed in Theorem \ref{thm:btt_explicit}.

Finally, we note that the above composition method is reminiscent of the classical technique of   \emph{code concatenation} \cite{For66}. Roughly speaking, code concatenation is a technique for reducing the alphabet size of a code, where one starts with a long outer code over a large alphabet, and then reduces the alphabet size by encoding each large alphabet symbol with a short inner code over a smaller alphabet. Curiously, the parameters obtained by concatenation are very similar to those obtained by the above Composition Lemma \ref{lem:composition}, when viewing $k$ as the codeword length and $q^m$ as the alphabet size of the code.

In particular, a well-known method for constructing \emph{asymptotically good} codes (i.e., codes with constant rate and distance) over small alphabets (e.g., the binary alphabet) is to first concatenate an inner asymptotically good code over a small alphabet of length $k_1\approx \log \log k$ (which can be found via brute-force search) with an outer $\rs$ code of alphabet size $\approx \exp(k_1)$ and  
length $k_2 \approx \log k \approx \exp(k_1)$, and then concatenate the resulting code with another outer $\rs$ code of alphabet size $\approx k  \approx \exp(k_2)$ and length $\approx k$. Our construction of $\btt$ evasive subspaces uses the same two-level construction, with the explicit subspace design of Theorem \ref{thm:subspace-design-explicit} playing the role of the $\rs$ code over a large alphabet, and the non-explicit $\btt$ evasive subspace of Lemma \ref{lem:btt_non_explicit} playing the role of the non-explicit asymptotically good code over a small alphabet. However, despite the technical resemblance, we could not find any formal connection between code concatenation and the above composition method for evasive subspaces.  We further note that a similar two-step composition approach was also used in other settings in theoretical computer science such as the original proof of the PCP theorem \cite{ALMSS98}.

\subsection{Proof of Main Theorem \ref{thm:main}}
Our main Theorem \ref{thm:main} follows as an immediate corollary of the above Theorems \ref{thm:intro_ag_btt} and \ref{thm:btt_explicit}.

\begin{proof}[Proof of Theorem \ref{thm:main}]
Let $\epsilon'=\frac{\epsilon} {26}$ and $R'=R+25\epsilon'$.
Let $m = 1/(\epsilon')^2$, and let $q=\poly(1/\epsilon)$ be an even power of a prime so that $q\geq m^c$, where $c>1$ is a sufficiently large constant for which both Theorems \ref{thm:intro_ag_btt} and \ref{thm:btt_explicit} hold. 
Let $\{C_n\}_{n}$ be the  infinite family of codes, guaranteed by Theorem \ref{thm:intro_ag_btt}, with the following properties  for each $C_n$:
\begin{enumerate}
\item $C_n: \F_{q^m}^k \to \F_{q^m}^n$ is a linear code of rate at least $R'$ that can be encoded in time $\poly(1/\eps,n)$.
\item There exists an injective linear map $\phi: \F_{q^m}^k \to \F_{q^m}^{\hat k}$, where $\hat k \leq n$, so that $C_n$ can be list decoded from a $(1-R'-\epsilon')$-fraction of errors, pinning down the images of the candidate messages under $\phi$ (viewed as length $\hat k m$ vectors over $\F_q$) to an affine shift of a $(\hat k,m,\epsilon' m)$-$\btt$ subspace $V$ over $\F_q$. Moreover, the map $\phi$, a basis for $V$, and the affine shift can be computed in time $\poly(1/\eps,n)$.
\end{enumerate}

Fix a code $C_n$ as above, and let $\epsilon'' = 25\epsilon'$.
By Theorem \ref{thm:btt_explicit}, there exists a $(\hat{k},m,\epsilon' m,s)$-$\btt$ evasive subspace $W \subseteq \F_q^{\hat k m}$  of co-dimension at most $ \epsilon'' \hat{k} m$ for $s= \poly(1/\epsilon)$, and  a basis for  $W$ can be found in time $\poly(1/\epsilon, n)$. Let $C'_n$ be the code obtained from $C_n$ by restricting the message space to $\phi^{-1}(W)$.
We claim  that $C'_n$ satisfies Theorem \ref{thm:main}.

Note first that $C'_n$ is an $\F_q$-linear code of 
rate at least $R' - \epsilon'' \hat{k} / n   \geq R' -  \epsilon'' = R$, and alphabet size $q^m=(1/\epsilon)^{O(1/\epsilon^2)}$, that can be encoded in time $\poly(1/\eps,n)$.
Moreover, as $1-R'-\epsilon' = 1 - R - \epsilon$, the code $C_n$ can be list decoded from a $(1-R-\epsilon)$-fraction of errors, 
pinning down the images of the candidate messages under $\phi$ (viewed as length $\hat k m$ vectors over $\F_q$) to an affine shift $\mathbf{u}$ of a $(\hat k,m,\epsilon' m)$-$\btt$ subspace $V$ over $\F_q$.
This means that the candidate messages of the code $C_n$ are contained in $\phi^{-1}(\mathbf{u} +V)$.

As $C'_n$ is obtained from $C_n$ by restricting the message space to $\phi^{-1}(W)$,
the candidate messages of $C'_n$ are contained in 
\[
\phi^{-1}(\mathbf{u}+V)\cap \phi^{-1}(W)=\phi^{-1}((\mathbf{u}+V) \cap W),
\] 
which is an affine shift of $\phi^{-1}(V\cap W)$ (or empty).
We can find a basis $B$ for $\phi^{-1}(V \cap W)$ and a vector $\mathbf{u}'\in \phi^{-1}((\mathbf{u}+V) \cap W)$ (if such exists)
in time $\poly(1/\epsilon, n)$ given the received word. Then the list of candidate messages of $C'_n$ is contained in the subspace spanned by $B$ and $\mathbf{u}'$ over $\F_q$, whose dimension is bounded by $\dim(V\cap W)+1\leq s+1=\poly(1/\epsilon)$, and a basis for this subspace can be found in time $\poly(1/\epsilon, n)$.
Consequently,  the output list size is $\exp(\poly(1/\epsilon))$ and the entire list can be output in time $\exp(\poly(1/\epsilon)) \cdot \poly(n)$, as claimed.
\end{proof}

\paragraph{Open problems.} 
We end this section with a couple of intriguing open problems.
\begin{enumerate}
\item Is it possible to explicitly construct capacity-achieving list decodable codes with list size $\poly(1/\epsilon)$ (even over a large super-constant alphabet)? As mentioned above, the smallest known list size for explicit capacity-achieving list decodable codes is $(1/\epsilon)^{O(1/\epsilon)}$, achieved by $\frs$ codes \cite{KRSW}, while potentially the list size could be as small as $O(1/\eps)$, as is the case for random codes. Such a construction could also potentially lead to \emph{fully polynomial-time} $\poly(1/\eps,n)$ list-decoding algorithms. 
\item Is it possible to obtain capacity-achieving list decodable codes with \emph{truly linear-time} encoding or list decoding algorithms?  As in the unique decoding setting, this seems to require completely different  techniques, e.g., graph-based constructions \cite{MRRSW, RWZ}. 
\item A question that is still widely open is to explicitly construct capacity-achieving list decodable codes over small \emph{fixed-size} alphabets, e.g., the \emph{binary} alphabet. Over a $q$-ary alphabet, the list-decoding capacity is known to be $h_q^{-1}(1-R)$, where $h_q (x)=x\log_q(q-1)+x\log_q(1/ x) +(1-x)\log_q(1/(1-x))$ is the $q$-ary entropy function. Once more, this question seems to require completely different techniques such as graph-based constructions \cite{Tashma17, GQST20}.
\item Can our methods be used to construct other pseudo-random objects? In particular, an intriguing question is whether these techniques could be used to construct \emph{lossless dimension expanders} over constant-size fields, whose state-of-the-art constructions \cite{GRX18} are based on the list-decoding machinery of \cite{GX13}.\end{enumerate}

\paragraph{Organization.} In Section \ref{sec:btt} we present our explicit construction of $\btt$ evasive subspaces (Theorem \ref{thm:btt_explicit}). In Section \ref{sec:RS} we first show, as a warm-up, that the output list of $\rs$ codes with subfield evaluation points is contained in an affine shift of a $\btt$ subspace. Then in Section \ref{sec:AG}, after providing the required $\ag$ code preliminaries in Section \ref{sec:AG_pre},
 we show how to extend the analysis to $\ag$ codes with subfield evaluation points (Theorem \ref{thm:intro_ag_btt}).




 

\section{Explicit BTT evasive subspace}\label{sec:btt}

In this section, we prove Theorem \ref{thm:btt_explicit}, which is restated below.

\bttexp*




The first ingredient in our proof is Lemma \ref{lem:btt_non_explicit}, restated below, which shows the existence of a (non-explicit) $(k,m,r,s)$-$\btt$ evasive subspace $W \subseteq \F_q^{km}$ of co-dimension at most $\epsilon km$ and $s = O\left(r/ \epsilon\right)$. Notably, the lemma holds for any field size $q$ and block length $m$.

\bttnonexp*

\begin{proof}
Let $W \subseteq \F_q^{km}$ be a random linear subspace of co-dimension $\epsilon k m$.   
Fix a $(k,m,r)$-$\btt$ subspace $V \subseteq \F_q^{km}$. 
We first bound the probability that $\dim(V \cap W) \geq s$. Fix a subspace $V^* \subseteq V$ of dimension $s$. Since $W$ is a random subspace of co-dimension $\epsilon k m$, the probability that $V^*$ is contained in $W$ is at most $\prod_{i=0}^{s-1}  \frac{q^{(1-\epsilon)km}-q^i}{q^{km}-q^i} \leq q^{-\epsilon k m s}$. As the number of $s$-dimensional subspaces $V^* \subseteq V$ is at most $q^{k r  s}$, by a union bound, we have that $W$ contains some $s$-dimensional subspace $V^* \subseteq V$ with probability at most $q^{(r-\epsilon m)ks}$. So $\dim(V \cap W) \geq s$ with probability at most $q^{(r-\epsilon m)ks}$. 

Next, we observe that the number of $(k,m,r)$-$\btt$ subspaces is at most $q^{r k m}$, as each such subspace is determined by the first $r$ columns of a $(k,m,r)$-$\btt$ matrix. Consequently, by a union bound, we get that $\dim(V \cap W) \geq s$ for some $(k,m,r)$-$\btt$ subspace $V$, with probability at most $q^{rkm + (r -\epsilon m)ks}$. This latter probability is smaller than $1$ since
$$r k m + (r-\epsilon m)ks =  r k m + (r-\epsilon m)\cdot \frac{2 kr} {\epsilon}= 
r k\left( \frac {2r} \epsilon -m \right) <0,$$
where the first equality is by our choice of $s=2 r / \epsilon$, and the last inequality is by our choice of 
$r< \frac{\epsilon m} {2}$. We conclude that there exists a $(k,m,r, \frac{2r} {\epsilon })$-$\btt$ evasive subspace $W$ of co-dimension at most $\epsilon km$.
\end{proof}

\begin{rem}\label{rem:btt_non_explicit}
We further note that to find a $\btt$ evasive subspace as above, one can enumerate over all subspaces $W \subseteq \F_q^{km}$ of co-dimension at most $\epsilon km$, and over all $(k,m,r)$-$\btt$ subspaces $V$, and compute the dimension of their intersection, which takes time $q^{O((km)^2)}$. 
\end{rem}

Our second ingredient is Theorem \ref{thm:subspace-design-explicit} from \cite{GK16}, restated below,  which gives an \emph{explicit} construction of 
an $(r,s)$-subspace design over $\F_{q^m}$ of cardinality $k$, where each subspace has  co-dimension at most $\epsilon m$ and $s = O\left(\frac{r^2} {\epsilon}\right)$, as long as $q^m$ is sufficiently larger than $k$. 

\sdexp*

For completeness, we sketch the proof of the above theorem in Appendix \ref{sec:subspace-design}. 
Note that the above theorem in particular gives an explicit periodic evasive subspace with the same parameters.

\begin{cor}[Explicit periodic evasive subspace]\label{cor:periodic-explicit}
There exists an absolute constant $c >1$, so that 
for every $\epsilon >0$, positive integers $k,m,r$ with $r < \frac{\epsilon m} {4}$, and a prime power $q$ satisfying
$q^m \geq \max\left\{k^{c \cdot r /\epsilon}, \left(\frac {2r} {\eps}\right) ^{2 r /\epsilon}\right\}$,
 there exists a $(k,m,r,s)$-periodic evasive subspace $W \subseteq \F_q^{km}$ of co-dimension at most $\eps k m$ for $s= \frac{2r^2} {\epsilon}$.  Moreover, a basis for  $W$ can be found in time $\poly(q,k,m)$.
\end{cor}

\begin{proof}
Let $H_1, H_2, \dots, H_{k}$ be the $(r,s)$-subspace design over $\F_{q^m}$ 
 guaranteed by Theorem \ref{thm:subspace-design-explicit} for the same choice of parameters, and let  $W=H_1\times H_2\times \cdots\times H_{k}$. Then by Theorem \ref{thm:subspace-design-explicit}, we clearly have that a basis for $W$ is a subspace of $\F_q^{km}$ of co-dimension at most $\eps k m$, and that $W$ can be found in time $\poly(q,k,m)$. It remains to show that $W$ is a $(k,m,r,s)$-periodic evasive subspace. Let $M$  be a $(k, m, r)$-periodic matrix, and let $V=\img(M)$, we would like to show that $\dim (V\cap W)\leq s$.

 By definition, $M$ is a block-lower-triangular matrix with $k$ copies of an $m\times r$ matrix  $\hat M$ on the main diagonal, and $\hat M$ has full column rank $r$.
 Let $\hat V=\img(\hat M)$, which is a subspace of $\F_q^m$ of dimension $r$.
 For $i\in [k]$, choose an $m\times m$ matrix $R_i$ such that $H_i=\ker(R_i)$, 
  and let $R \in \F_q^{km \times km}$ be a $(k\times k)$-block diagonal matrix with blocks $R_1, R_2,\dots, R_k$ on the main diagonal. Note that $W=\ker(R)$, and furthermore, $R M \in \F_q^{km \times kr}$ is a $(k \times k)$-block-lower-triangular matrix with blocks $R_1 \hat M, R_2 \hat M,\dots, R_k \hat M $ on the main diagonal.

So we have
\begin{align*}
 \dim (V \cap W)& = \dim (V \cap \ker(R)) \\
 &= \dim (\ker(RM)) \leq \sum_{i=1}^{k} \dim (\ker(R_i \hat M)) \\
&=\sum_{i=1}^{k} \dim (\hat V\cap \ker(R_i))=\sum_{i=1}^{k} \dim (\hat V\cap H_i) \leq s, 
\end{align*}
where the  last inequality follows since $H_1,\ldots,H_k$ is an $(r,s)$-subspace design.
\end{proof}

Our last ingredient is Lemma \ref{lem:composition}, restated below, which gives a composition lemma for $\btt$ evasive subspaces. 

\composition*

\begin{proof}
Let $V$ be a ($k' k,m, r$)-$\btt$ subspace. Our goal is to show that $\dim(V \cap (W^{k'} \cap W')) \leq s'$. Since $W'$ is a $(k', k m, s, s')$-periodic evasive subspace, it suffices to show that $V':=V \cap W^{k'}$ is contained in a $(k', km, s)$-periodic subspace $U$, and consequently 
$$\dim(V \cap (W^{k'} \cap W'))=\dim((V \cap W^{k'}) \cap W') \leq \dim(U \cap W') \leq s' .$$

Since $V$ is a ($k'k, m, r$)-$\btt$ subspace, there exists a ($k'k, m,r$)-$\btt$ matrix $M$ whose image equals $V$. As a $ k'\times k'$ block matrix, $M$ has $k'$ copies of a $(k m) \times (k r)$ block $\hat M$ on its main diagonal. 
Next observe that $\hat M$ itself is a $(k,m,r)$-$\btt$ matrix, and so if we let $\hat V:=\img(\hat M)$, then $\hat V$ is a $(k,m,r)$-$\btt$ subspace. Recalling our assumption that $W$ is a $(k,m,r,s)$-$\btt$ evasive subspace, this implies in turn that  $\dim(\hat V \cap W) \leq s$. Let $H \subseteq \F_q^{km}$ be a subspace of dimension $s$ containing $\hat V \cap W$, and let $\{b^{(1)},\ldots, b^{(s)}\} \subseteq \F_q^{km}$ be a basis for $H$.

Next, we introduce a bit of notation. We write a vector $ x\in\F_q^{k'km}$ as $x=(x_1,x_2,\dots,x_{k'})$ where $x_i\in \F_q^{km}$, and for $i = 1,\ldots,k'$, we let $\pi_i(x) = x_i$.
 For a subspace $X \subseteq \F_q^{k'km}$,
and $i=0,1,\dots, k'$, we let
$$X_i=\left\{ (x_1,x_2,\dots,x_{k'}) \in X \mid x_1=x_2 = \cdots =x_{i}=0 \right\}.$$
In particular, we have $X_0=X$ and $X_k=\{0\}$. 

\begin{claim}\label{claim_spanning}
For all $i=1,\dots, k'$,  there exist vectors $b^{(i,1)}, \dots, b^{(i,s)}\in \F_q^{k'km}$ satisfying the following conditions:
\begin{enumerate}
\item  For all $j=1,\dots,s$, it holds that $(b^{(i,j)})_1 = \cdots =(b^{(i,j)})_{i-1} = 0$ and $(b^{(i,j)})_i=b^{(j)}$.
\item $(V \cap W^{k'})_{i-1}\subseteq (V \cap W^{k'})_i + \mathrm{span}\{b^{(i,1)}, \dots, b^{(i,s)}\}$.
\end{enumerate}
\end{claim}
\begin{proof}
Fix $i\in [k']$, and let $H'=\pi_i((V \cap W^{k'})_{i-1})$. Note that 
\[
H'\subseteq \pi_i(V_{i-1}) \cap \pi_i((W^{k'})_{i-1})=\hat{V}\cap W\subseteq H.
\]
Let $t=\dim H'$.
Fix a basis $\{v^{(1)},\dots, v^{(t)}\}$ for $H'$ and extend it to a basis $\{v^{(1)},\dots, v^{(s)}\}$ for $H$.
For $j=1,\dots, t$,  choose $u^{(j)}\in (V \cap W^{k'})_{i-1}$ such that $(u^{(j)})_i=v^{(j)}$,
which is possible since $v^{(j)}\in H'=\pi_i((V \cap W^{k'})_{i-1})$.
For $j=t+1,\dots, s$, choose $u^{(j)}\in \F_q^{k'km}$ with $(u^{(j)})_1 = \cdots = (u^{(j)})_{i-1}=0 $ such that $(u^{(j)})_i=v^{(j)}$.
Then we have
\begin{equation}\label{eq_span1}
(V \cap W^{k'})_{i-1} = (V \cap W^{k'})_{i} + \mathrm{span}\{u^{(1)}, \dots, u^{(t)}\} \subseteq (V \cap W^{k'})_{i} + \mathrm{span}\{u^{(1)}, \dots, u^{(s)}\}.
\end{equation}

As $\{b^{(1)},\dots, b^{(s)}\}$ and $\{v^{(1)},\dots, v^{(s)}\}$ are both bases of $H$, there exists a unique invertible $s\times s$ matrix $A=(a_{j,\ell})_{j,\ell\in [s]}$
over $\F_q$ such that $b^{(j)}=\sum_{\ell=1}^s a_{j,\ell} v^{(\ell)}$ for $j\in [s]$.
For $j\in [s]$, let $b^{(i,j)}=\sum_{\ell=1}^s a_{j,\ell} u^{(\ell)}$, and note that $(b^{(i,j)})_1 = \cdots =(b^{(i,j)})_{i-1} = 0$
 and 
\[
(b^{(i,j)})_i=\sum_{\ell=1}^s a_{j,\ell} (u^{(\ell)})_i=\sum_{\ell=1}^s a_{j,\ell} v^{(\ell)}=b^{(j)}.
\]
So the first condition of the claim is satisfied.
As $A$ is invertible, we have $\mathrm{span}\{b^{(i,1)}, \dots, b^{(i,s)}\}=\mathrm{span}\{u^{(1)}, \dots, u^{(s)}\}$. Combining this with \eqref{eq_span1} proves the second condition.
\end{proof}

Now recall that our goal is to exhibit a $(k',km,s)$-periodic matrix $\tilde M$ so that $V \cap W^{k'} \subseteq \img(\tilde M)$. We construct $\tilde M$ as follows. 
For $i=1,\ldots,k'$, let $M_i$ be a $(k'km) \times s$ matrix whose columns are $b^{(i,1)}, \ldots, b^{(i,s)}$. 
Let
\[
\tilde M=\begin{pmatrix} M_1  & M_2  & \cdots  & M_{k'}\end{pmatrix}.
\] 
By the first condition of Claim~\ref{claim_spanning}, we have that $\tilde M$ is a $(k',km,s)$-periodic matrix. By the second condition of  Claim~\ref{claim_spanning}, 
we further have that $(V \cap W^{k'})_{i-1}\subseteq (V \cap W^{k'})_i + \img(M_i)$ for all $i=1,\ldots k'$, and so $V \cap W^{k'} = (V \cap W^{k'})_0\subseteq \img(\tilde M)$, as claimed. This completes the proof of the lemma.
\end{proof}


Next, we prove Theorem \ref{thm:btt_explicit} based on the ingredients above.

\begin{proof}[Proof of Theorem \ref{thm:btt_explicit}]

Our goal is to construct a $(k,m,r,s)$-$\btt$ evasive subspace $W$ of co-dimension at most $\epsilon k m$ for $s= \poly(r/\epsilon)$. 
We shall construct $W$ by applying two composition steps. In the first step, we shall compose an inner $\btt$ evasive subspace $W_1$, given by Lemma \ref{lem:btt_non_explicit}, which can be found via brute-force search, with
an outer explicit periodic evasive subspace $W_2$, given by  Corollary~\ref{cor:periodic-explicit},
 to obtain a $\btt$ evasive subspace $W_1 \circ W_2$. In the second step, we shall compose the resulting $\btt$ evasive subspace $W_1 \circ W_2$ with yet another outer explicit periodic evasive subspace $W_3$, given by  Corollary~\ref{cor:periodic-explicit}, to obtain our final $\btt$ evasive subspace $W: = (W_1 \circ W_2) \circ W_3$.
One technical issue is that the desired number of blocks $k$ may not be a multiple of the number of blocks of the inner subspace.
This is solved by first constructing a $\btt$ evasive subspace $W'$ in a slightly larger ambient space $\F_q^{k'm}\supseteq \F_q^{km}$ and then letting $W= W'\cap \F_q^{km}$ (where $\F_q^{km}$ is identified with a subspace of $\F_q^{k'm}$ via the map $(x_1,\dots,x_{km})\mapsto (0, \dots, 0, x_1,\dots,x_{km})$).

In the following, assume $c>1$ is a large enough constant.
Let $\epsilon' = \epsilon / 6$. By assumption, we have $r<\epsilon m/24 = \epsilon' m/4$ and $q\geq m^c$.

\paragraph{BTT evasive subspace $W_1$:} Let $W_1$ be a $( k_1,m,r, s_1 )$-$\btt$ evasive subspace of co-dimension at most $\epsilon' k_1 m$ for $k_1 = c^2 m^3\cdot \lceil \frac{\log \log k} {\log q}\rceil$  
and $s_1=\frac{2r^2} {\epsilon'}\geq \frac{2r} {\epsilon'}$.
Note that such a subspace exists by Lemma \ref{lem:btt_non_explicit}. 
We further claim that a basis for $W_1$ can be found in time $\poly(q,k,m)$. 
To see this, first note that if $m\log q\leq (c^2\log\log k+c^2)^{c+1}$, then by Remark \ref{rem:btt_non_explicit}, a basis for $W_1$ can be found in time 
\begin{equation}\label{eq_k1time}
q^{O( (k_1 \cdot m)^2)} = 2^{O(k_1^2 m^2 \log q)} \leq \exp ( \poly(\log \log k)) \leq \poly(k).
\end{equation}

On the other hand, if $m\log q > (c^2\log\log k+c^2)^{c+1}$, then either $m>c^2\log\log k+c^2$ or $q>\log q>(c^2\log\log k+c^2)^c$. 
In either case, we have $q\geq (c^2\log\log k+c^2)^c$ since $q\geq m^c$.
This implies
\[
q\geq q^{1/4} m^{(3/4)c}>\left(c^2 m^3\cdot\left \lceil \frac{\log \log k} {\log q}\right\rceil\right)^{c/4} = k_1^{c/4}.
\] 
Therefore,
\begin{equation}\label{eq_k1}
q^m\geq  \max\{k_1^{cm/4}, m^m\}\geq \max\left\{k_1^{c \cdot r /\epsilon'}, \left(\frac{2r}{\epsilon'}\right)^{2r/\epsilon'}\right\}
\end{equation}
where we use the facts  $r<\epsilon' m/4$ and $q\geq m^c\geq m$.
 Consequently,  by Corollary~\ref{cor:periodic-explicit},  there exists a $( k_1,m,r, \frac{2r^2} {\epsilon'})$-periodic evasive subspace $W_1$ of  co-dimension at most $ \epsilon' k_1 m$, which is in particular a $\btt$ evasive subspace with the same parameters. Moreover, a basis for $W_1$ can be found in
 time $\poly(q,k_1,m)= \poly(q,k,m)$.

\paragraph{Periodic evasive subspace $W_2$:} Let $W_2$ be a $( k_2, k_1 m, s_1, s_2)$-periodic evasive subspace of co-dimension at most $\epsilon' k_2 k_1 m$ for $k_2 = \lceil \log k\rceil$ and $s_2 = \frac{ 2s_1^2} {\epsilon'} = \frac {8 r^4} {(\epsilon')^3}$. Note that such a subspace exists by Corollary~\ref{cor:periodic-explicit} as
\[
s_1=\frac{2r^2} {\epsilon'}<\frac{c^2 \epsilon' m^2} {4}\leq \frac{\epsilon' k_1 m} {4},
\]
\begin{equation}\label{eq_k2}
q^{k_1 m}  \geq   \lceil \log k\rceil^{c^2 m^2} >  \lceil \log k\rceil^{2c (r/\epsilon')^2}=k_2^{c \cdot s_1 /\epsilon'},
\end{equation}
and
\[
q^{k_1 m}\geq  \left(\frac{m^2}{4}\right)^{m^2/4} > \left(\frac{4r^2}{(\epsilon')^2}\right)^{4r^2/(\epsilon')^2} = \left(\frac{2s_1}{\epsilon'}\right)^{2s_1/\epsilon'}
\]
where the inequalities hold  by  the choice of $k_1 = c^2 m^3\cdot \lceil \frac{\log \log k} {\log q} \rceil$  and the assumptions  $r < \frac{\epsilon' m} {4}$ and $q\geq m^c$. Moreover, a basis for $W_2$ can be found in
 time $\poly(q,k_2,k_1m)= \poly(q,k,m)$.

\paragraph{BTT evasive subspace $W_1 \circ W_2$:} By Lemma \ref{lem:composition}, we have that $W_1 \circ W_2 = W_1^{k_2}\cap W_2$ is 
a $(k_1k_2, m,r,s_2)$-$\btt$ evasive subspace for $s_2 =  \frac {8 r^4} {(\epsilon')^3}$. Note furthermore that $W_1 \circ W_2$ has  co-dimension at most $2\epsilon' k_1 k_2 m $, and a basis for $W_1 \circ W_2$ can be found in time $\poly(q,k,m)$.

\paragraph{Periodic evasive subspace $W_3$:}
Let $W_2$ be a $( k_3, k_1 k_2 m, s_2, s_3)$-periodic evasive subspace of co-dimension at most $ \epsilon' k_1 k_2 k_3 m$ for  $k_3 = \lceil \frac {k} {k_1 k_2}\rceil $  and $s_3 = \frac{ 2s_2^2}{\epsilon'} = \frac { 128 r^8} {(\epsilon')^7}$. Note that such a subspace exists by Corollary~\ref{cor:periodic-explicit} as
\[
s_2= \frac {8 r^4} {(\epsilon')^3}<\frac{c^2 \epsilon' m^4} {4}\leq \frac{\epsilon' k_1 k_2 m} {4},
\]
\[
q^{k_1 k_2 m}   \geq k^{c^2 m^4 }> k^{8 c (r/\epsilon')^4} \geq  k_3^{c \cdot s_2 /\epsilon'},
\]
 and
 \[
 q^{k_1 k_2 m}   \geq \left(\frac{m^4}{16}\right)^{m^4/16} > \left(\frac{16r^4}{(\epsilon')^4}\right)^{16r^4/(\epsilon')^4} = \left(\frac{2s_2}{\epsilon'}\right)^{2s_2/\epsilon'}
 \]
where the inequalities hold  once more by the choice of $k_1 = c^2 m^3\cdot \lceil \frac{\log \log k} {\log q} \rceil$ and $k_2 = \lceil\log k\rceil$ together with the assumptions $r<\epsilon' m/4$ and $q\geq m^c$. Moreover, a basis for $W_3$ can be found in
 time $\poly(q,k_3,k_1k_2m)= \poly(q,k,m)$.
 
\paragraph{BTT evasive subspace $W'=(W_1 \circ W_2) \circ W_3$:}
By Lemma \ref{lem:composition}, we have that $W':=(W_1 \circ W_2) \circ W_3$ is 
a $(k_1 k_2 k_3, m,r,s_3)$-$\btt$ evasive subspace for $s_3 =  \frac {128  r^8} {\epsilon^7} = \poly(r/\epsilon)$. Note furthermore that $W'$ has co-dimension at most $3\epsilon' k_1 k_2 k_3 m$, and a basis for $W'$ can be found in time $\poly(q,k,m)$.

\paragraph{BTT evasive subspace $W$:}

If $k$ is a multiple of $k_1 k_2$, then $k_1 k_2 k_3=k$ and we may choose $W=W'$ as the desired $\btt$ evasive subspace.
Next, we explain how to extend it to arbitrary $k$.

\begin{enumerate}
\item First assume $k\geq k_1 k_2$ so that $k':=k_1 k_2 k_3=k_1 k_2 \lceil \frac {k} {k_1 k_2}\rceil$ satisfies $k'  \leq 2k$.
Then $W'\subseteq\F_q^{k'm}$ is a  $(k', m, r, s_3)$-$\btt$ evasive subspace of co-dimension at most $3\epsilon' k' m\leq 6\epsilon' km=\epsilon km$.
Identify $\F_q^{km}$ with a subspace of $\F_q^{k'm}$ via the map $(x_1,\dots,x_{km})\mapsto (0, \dots, 0, x_1,\dots,x_{km})$.
We let $W:=W'\cap \F_q^{km}$, whose co-dimension in $\F_q^{km}$ is at most $\epsilon km$ since the co-dimension of $W'$ in $\F_q^{k'm}$ is at most $\epsilon km$.

Consider any $(k, m, r)$-$\btt$ subspace $V\subseteq \F_q^{km}$. Note that there exists a $(k', m, r)$-$\btt$ subspace $V'\subseteq \F_q^{k'm}$ such that $V=V'\cap \F_q^{km}$.  
As $W'$ is a $(k', m, r, s_3)$-$\btt$ evasive subspace, we have $\dim(V'\cap W')\leq s_3$, which implies $\dim(V\cap W)\leq s_3$.
So $W\subseteq \F_q^{km}$ is a $(k, m, r, s)$-$\btt$ evasive subspace of co-dimension at most $\epsilon km$ for $s=s_3=\poly(r/\epsilon)$.

\item Now assume $k_1\leq k<k_1 k_2$. Let $k'_2:=\lceil \frac{k}{k_1}\rceil\leq k_2$ so that $k_2' k_1\leq 2k$.
By replacing $k_2$ with $k_2'$ in the construction of $W_2$, we may construct a  $(k_2', k_1 m, s_1, s_2)$-periodic evasive subspace $W_2'$ of co-dimension at most $\epsilon' k_2' k_1 m$. 
This is because replacing $k_2$ by $k_2'\leq k_2$ preserves \eqref{eq_k2}.
Composing $W_2'$ with $W_1$ gives a $(k_1 k_2', m,r,s_2)$-$\btt$ evasive subspace $W_1\circ W_2'\subseteq\F_q^{k_1k_2' m}$  of co-dimension at most $2\epsilon' k_2' k_1 m\leq 4\epsilon' km\leq \epsilon k m$. Similarly to the previous case, restricting to the subspace $\F_q^{k m}$ yields the desired $(k, m,r,s)$-$\btt$ evasive subspace $W$ of co-dimension at most $\epsilon k m$ for $s=s_2=\poly(r/\epsilon)$.

\item Finally, assume $k<k_1$.
By replacing $k_1$ with $k$ in the construction of $W_1$, we may construct the desired $(k, m, r, s)$-$\btt$ evasive subspace $W$ of co-dimension at most $\epsilon' k m\leq \epsilon km$ for $s=s_1=\poly(r/\epsilon)$. 
This is because  replacing $k_1$ by $k<k_1$ preserves \eqref{eq_k1time} and \eqref{eq_k1}.
\end{enumerate}
\end{proof}




  \section{Reed--Solomon codes with subfield evaluation points}\label{sec:RS}
  
We first show, as a warm-up, that $\rs$ codes with evaluation points over a subfield are list decodable up to capacity with the output list contained in an affine shift of a $\btt$ subspace.  Later, in Section \ref{sec:AG}, we shall show how the analysis can be extended to $\ag$ codes over constant-size alphabets, thus proving Theorem \ref{thm:intro_ag_btt}. We start with the formal definition of $\rs$ codes with subfield evaluation points.

\begin{defi}[$\rs$ codes with subfield evaluation points]\label{defn:RS-subfield}
  Let $n, k,m\in\N^+$ be such that $k\leq n$, and let $q\geq n$ be a prime power. 
The \emph{Reed--Solomon code $\mathsf{RS}_{q,m}(n,k)$ over $\F_{q^m}$ with evaluation points in $\F_q$} maps
a polynomial $f \in \F_{q^m}[X]$ of degree at most $k-1$ (viewed as a length $k$ vector of coefficients over $\F_{q^m})$ to the codeword
$C_f:=(f(\alpha_1), f(\alpha_2), \dots, f(\alpha_n)) \in (\F_{q^m})^n$,
  where $\alpha_1, \alpha_2, \dots, \alpha_n$ are $n$ distinct elements in $\F_q$.
\end{defi}

Note that $\mathsf{RS}_{q,m}(n,k)$ is a linear code over the alphabet $\F_{q^m}$  with codeword length $n$,  rate $k/n$, and minimum distance $n-k+1$. In this section, we show that this code is also list decodable up to its minimum distance with the output list being contained in an affine shift of a $\btt$ subspace.

\begin{thm}\label{thm_rs}
  Let $\epsilon >0$, let $n, k,m\in\N^+$ be such that $k\leq n$ and $m \geq 1/\epsilon^2$, and let $q\geq n$ be a prime power. 
  Then $\mathsf{RS}_{q,m}(n,k)$ can be list decoded from up to $(1-\epsilon)(n-k)$ errors, pinning down the candidate messages (viewed as length $km$ vectors of coefficients over $\F_q$)
  to an affine shift of a $(k,m,\epsilon m)$-$\btt$ subspace $V$ over $\F_q$. Moreover, a basis for $V$ and the affine shift can be found in time $\poly(\log q,m,n)$.
\end{thm}

The above theorem is a consequence of the following lemma.

\begin{lem}\label{lem_rs}
  Let $n, k,m\in\N^+$ be such that $k\leq n$, and let $q\geq n$ be a prime power. Let $s\in [m]$ and $t, d\in\N^+$ be parameters, 
 satisfying that
 \begin{equation}\label{eq_rs_cond1}
(s+1)(d+1)+k-1>n
\end{equation}
and
\begin{equation}\label{eq_rs_cond2}
 t>d+k-1.
\end{equation}
  Then $\mathsf{RS}_{q,m}(n,k)$ can be list decoded from agreement at least $t$, pinning down the candidate messages (viewed as length $km$ vectors of coefficients over $\F_q$) to  an affine shift of a $(k,m,s-1)$-$\btt$ subspace $V$ over $\F_q$. Moreover, a basis for $V$ 
   and the affine shift
  can be found in time $\poly(\log q,m,n)$. 
\end{lem}

Before we prove the above lemma, we show how it implies Theorem \ref{thm_rs}.

 \begin{proof}[Proof of Theorem \ref{thm_rs}]
Let $s= \frac 1\epsilon +1$, let $d=  \frac{n-k+2}{s+1}-1$ so that \eqref{eq_rs_cond1} is satisfied, and 
let $t=d+k$ so that \eqref{eq_rs_cond2} is satisfied. Then with this setting of parameters, $\mathsf{RS}_{q,m}(n,k)$ can be list decoded from agreement  $t$, or equivalently, from up to
$$
n-t = n-d-k = n- k +1 - \frac{n-k+2}{s+1} \geq n- k- \frac{n-k}{s-1} =
(1-\epsilon)(n-k) .
$$
errors.
Moreover, by choice of $m \geq 1/\epsilon^2$, we have that $V$ is a $(k,m,s-1)$-$\btt$ subspace for $s-1 = \frac 1 \epsilon \leq \epsilon m$.
\end{proof}

The rest of this section is devoted to the proof of Lemma \ref{lem_rs}.
To prove this lemma, we follow the linear-algebraic approach of \cite{GX13}. Suppose that $\mathbf{y}=(y_1,y_2,\dots,y_n)\in (\F_{q^m})^n$ is a received word. Our goal is to show that all polynomials $f\in \F_{q^m}[X]_{<k}$ that have large agreement with $\mathbf{y}$ are contained in an affine shift of a $\btt$ subspace.
To this end, following \cite{GX13}, we first show in Section \ref{subsec_rs_Q}
the existence of a nonzero polynomial $Q$  (depending on the received word $\mathbf{y}$) that gives a functional equation that any polynomial $f$ that has large agreement with $\mathbf{y}$ must satisfy. Then we show in Section 
\ref{subsec_rs_V}
that this functional equation induces a special structure on the solution set, specifically, the solution set is contained in an affine shift of the kernel of a $\btt$ matrix. Finally, in Section \ref{subsec_btt_ker_equiv} we show that the kernel of a $\btt$ matrix is a $\btt$ subspace, which implies that the solution set is contained in an affine shift of a BTT subspace.
We further show that the polynomial $Q$, a basis for  the $\btt$ subspace $V$, and the affine shift  could be found efficiently. 

\subsection{The polynomial $Q$}\label{subsec_rs_Q}

In what follows, let $\sigma\in\gal(\F_{q^m}/\F_q)$ be the \emph{Frobenius automorphism} $a\mapsto a^q$ of $\F_{q^m}$ over $\F_q$. It extends to an automorphism of $\F_{q^m}[X]$ over $\F_q$ by acting on the coefficients, which we also denote by $\sigma$ by a slight abuse of notation.
For $f\in\F_{q^m}[X]$, denote by $f^\sigma$ the element $\sigma(f)$.

Suppose that $\mathbf{y}=(y_1,y_2,\dots,y_n)\in (\F_{q^m})^n$ is a received word. 
We let $Q$ be a nonzero multivariate polynomial in
$(\F_{q^m}[X])[Y_1, Y_2, \dots, Y_s]$ of the form
\[
Q=A_0 + A_1 Y_1 + A_2 Y_2 + \dots + A_s Y_s,
\]
where $A_0, A_1, \dots, A_s\in \F_{q^m}[X]$, $\deg(A_0)\leq d+k-1$, and $\deg(A_i)\leq d$ for $i=1, 2, \dots, s$.
We also require the polynomials $A_i$ to satisfy the constraint
\begin{equation}\label{eq_rs_constraint}
A_0(\alpha_i) + A_1(\alpha_i) y_i + A_2(\alpha_i) y_i^\sigma + \dots + A_s(\alpha_i) y_i^{\sigma^{s-1}}=0
\end{equation}
for all $i = 1,\ldots,n$, where $\alpha_1,\dots,\alpha_n$ are the evaluation points.

We first claim that such a nonzero polynomial $Q$ exists and can be computed efficiently. To see this, think of the coefficients of the polynomials $A_i$ as unknowns. This gives $d+k+s (d+1)=(s+1)(d+1)+k-1$ unknowns in total, while \eqref{eq_rs_constraint} gives $n$ homogeneous linear constraints over $\F_{q^m}$.
 By \eqref{eq_rs_cond1}, the number of unknowns is greater than the number of linear constraints which guarantees the existence of a nonzero solution $Q$. Moreover, we can  find $Q$ in time $\poly(\log q ,m,n)$ 
 by solving the system of linear equations represented by \eqref{eq_rs_constraint}.

Next, we show that $Q$ gives a functional equation that any $f$ that has sufficiently large agreement with the received word $\mathbf{y}$ needs to satisfy.

\begin{claim}\label{clm_rs_funceqn}
Let $f\in \F_{q^m}[X]_{<k}$. Suppose that $\mathbf{y}$ agrees with the codeword $C_f=(f(\alpha_1), f(\alpha_2), \dots, f(\alpha_n))$ in at least $t$ coordinates.
Then $f$ satisfies the functional equation
\begin{equation}\label{eq_rs_funceqn}
Q(f, f^{\sigma},\dots, f^{\sigma^{s-1}})=A_0+A_1 f+ A_2 f^{\sigma}+\dots+A_s f^{\sigma^{s-1}}=0.
\end{equation}
\end{claim}

\begin{proof}
Define
\[
Q^*=A_0+A_1 f+ A_2 f^{\sigma}+\dots+A_s f^{\sigma^{s-1}}\in\F_{q^m}[X].
\]
We want to prove that $Q^*=0$.
As $\deg(f)\leq k-1$, $\deg(A_0)\leq d+k-1$, and $\deg(A_i)\leq d$ for $i=1, 2, \dots, s$, we know that $\deg(Q^*)\leq d+k-1$.

Suppose that $\mathbf{y}$ agrees with $C_f$ in the $i$-th symbol for some $i\in [n]$,   i.e., $y_i=f(\alpha_i)$.
By \eqref{eq_rs_constraint}, we have
\begin{align*}
0&=A_0(\alpha_i) + A_1(\alpha_i) y_i + A_2(\alpha_i) y_i^\sigma + \dots + A_s(\alpha_i) y_i^{\sigma^{s-1}}\\
&=A_0(\alpha_i) + A_1(\alpha_i) f(\alpha_i) + A_2(\alpha_i) (f(\alpha_i))^\sigma + \dots + A_s(\alpha_i) (f(\alpha_i))^{\sigma^{s-1}}\\
&=A_0(\alpha_i) + A_1(\alpha_i) f(\alpha_i) + A_2(\alpha_i) f^\sigma(\alpha_i) + \dots + A_s(\alpha_i) f^{\sigma^{s-1}}(\alpha_i)\\
&=(A_0+A_1 f+ A_2 f^{\sigma}+\dots+A_s f^{\sigma^{s-1}})(\alpha_i)\\
&=Q^*(\alpha_i).
\end{align*}
The third equality uses the fact that $(f(\alpha_i))^\sigma=f^\sigma(\alpha_i^\sigma)=f^\sigma(\alpha_i)$, which holds since $\alpha_i\in\F_q$ is fixed by $\sigma$.

As  $\mathbf{y}$ and $C_f$ agree in at least $t$ symbols, the above argument shows that $Q^*$ has at least $t$ zeros.
On the other hand, the degree of $Q^*$ is at most $d+k-1$, which is less than $t$ by \eqref{eq_rs_cond2}. This implies $Q^*=0$.
\end{proof}

\subsection{The BTT subspace $V$}\label{subsec_rs_V}
Next, we show that the functional equation \eqref{eq_rs_funceqn}, given by Claim \ref{clm_rs_funceqn} above, implies that the list of candidate messages is contained in an affine shift of the kernel of a $\btt$ subspace. We start by expanding the functional equation \eqref{eq_rs_funceqn} in terms of the coefficients of the polynomial $f$ and $A_0,A_1,\ldots,A_s$.

As $f\in \F_{q^m}[X]_{<k}$, we may write
\[
f=\sum_{j=0}^{k-1} f_j X^j
\] 
where the coefficients $f_i$ are in $\F_{q^m}$.
Also write
\[
A_\ell=\sum_{i=0}^{d+k-1} a_{\ell,i} X^i, \quad \ell=0,1,\dots,s,
\] 
where the coefficients $a_{\ell,i}$ are in $\F_{q^m}$ and $a_{\ell,i}=0$ for $\ell\in [s]$ and $i>d$.
Choose the largest integer $u\geq 0$ such that $X^u$ divides $A_\ell$ for $\ell=0,1,\dots, s$.
By replacing $Q$ with $Q/X^u$, we may assume that $u=0$. So $a_{\ell,0}\neq 0$ for some $\ell\in\{0,1,\dots,s\}$.
By \eqref{eq_rs_funceqn}, we actually have  $a_{\ell,0}\neq 0$ for some $\ell\in [s]$.

With the notations above, \eqref{eq_rs_funceqn} becomes
\begin{align*}
0&=\sum_{i=0}^{d+k-1} a_{0,i} X^i +\sum_{\ell=1}^s \left( \sum_{i=0}^{d} a_{\ell,i} X^i \right) \left(  \sum_{j=0}^{k-1} f_j^{\sigma^{\ell-1}} X^j\right) \\
&=  \sum_{i=0}^{d+k-1}\left(a_{0,i}+\sum_{\ell=1}^s \sum_{j=0}^{i} a_{\ell,i-j} f_{j}^{\sigma^{\ell-1}} \right) X^i, 
\end{align*}
where we let $f_i=0$ for $i\geq k$.
So we obtain the equations
\begin{equation}\label{eq_rs_eqns}
 \sum_{j=0}^{i} \sum_{\ell=1}^s a_{\ell,i-j} f_{j}^{\sigma^{\ell-1}}  = - a_{0,i}, \qquad i=0,1, \dots, k-1.
\end{equation}

Next, we show that the solution set of all $f =(f_0,f_1,\ldots,f_{k-1})$ satisfying \eqref{eq_rs_eqns} is contained in an affine shift of the kernel of a $\btt$ matrix. 
In what follows, fix an arbitrary $\F_q$-linear bijection $\phi: \F_{q^m} \to \F_q^m$. For an element $a \in \F_{q^m}$, let $\bar a:=\phi(a) \in \F_q^m$, and for a vector
$\mathbf{f} =(f_0,f_1,\ldots,f_{k-1}) \in (\F_{q^m})^k$, let $\bar {\mathbf{f}}:=(\bar f_0, \bar f_1,\ldots, \bar f_{k-1}) \in \F_q^{mk}$.

\begin{claim}\label{clm_rs_btt_ker}
Let $S$ be the set of all vectors $f =(f_0,f_1,\ldots,f_{k-1}) \in (\F_{q^m})^k$ satisfying that 
\begin{equation}\label{eq_rs_eqns_zero}
 \sum_{j=0}^{i} \sum_{\ell=1}^s a_{\ell,i-j} f_{j}^{\sigma^{\ell-1}}  =0, \qquad i=0,1, \dots, k-1,
\end{equation}
where  $a_{\ell,i-j} \in \F_{q^m}$, and $a_{\ell,0}\neq 0$ for some $\ell\in [s]$.
Let $\bar S:= \{\bar f \mid f \in S\} \subseteq \F_q^{mk}$.
Then  $\bar S \subseteq \ker(M)$ for a $(k,r,m)$-$\btt$ matrix $M$ over $\F_q$ with $m-s+1\leq r\leq m$. Moreover, $M$ can be constructed in time $\poly(\log q,m,n)$.
\end{claim}

\begin{proof}
First note that both the multiplication map $m_b: \F_{q^m}\to\F_{q^m}$, given by
$a \mapsto b \cdot a$ for $a,b \in \F_{q^m}$, and the Frobenius automorphism $\sigma: \F_{q^m}\to\F_{q^m}$, given by $a \mapsto a^q$ for $a \in \F_{q^m}$, are $\F_q$-linear operations over $\F_{q^m}$. 
Consequently, for all $i=0,1,\ldots,k-1$,  there exists an $m \times m$ matrix $M_i$ over $\F_q$
so that 
$
\sum_{\ell=1}^s a_{\ell,i} \cdot b^{\sigma^{\ell-1}} = M_i \cdot \bar b
$
 for every $b \in \F_{q^m}$.
In this notation, we can rewrite \eqref{eq_rs_eqns_zero} as $M \cdot \bar f =0$, where 
\[M=
\begin{pmatrix}
M_{0} & 0 & 0 & \cdots & 0  \\
M_{1} & M_{0} & 0 & \cdots & 0 \\
M_{2} & M_{1} & M_{0} & \cdots & 0 \\
\vdots  & \vdots  & \vdots & \ddots & \vdots  \\
M_{k-1} & M_{k-2} & M_{k-3} & \cdots & M_{0}
\end{pmatrix}.
\]

Then we have that $M$ is a block lower-triangular Toeplitz matrix with blocks of size $m \times m$. To obtain a $\btt$
 matrix, we need to further ensure that all matrices $M_0$ have full rank. For this, we let $r:=\rank(M_0)$, and choose a subset $B$ of $r$ linearly independent rows of $M_0$. Then in the matrix $M$, we only keep the rows whose projection on the block $M_0$ belongs to $B$. This clearly gives a $(k,r,m)$-$\btt$ matrix $M$ so that $\bar S \subseteq \ker(M)$. 
Moreover, $M$ can clearly be constructed in time $\poly(\log q,m,n)$.
To conclude the proof of the claim, it remains to show that $r=\rank(M_0) \geq m-s+1$.

To see that $\rank(M_0) \geq m-s+1$, we show that $\dim(\ker(M_0)) \leq s-1$. Recall that $M_0$ represents the $\F_q$-linear map 
$b \mapsto \sum_{\ell=1}^s a_{\ell,0} \cdot b^{\sigma^{\ell-1}}$ for $b \in \F_{q^m}$.
 Recalling our assumption that  $a_{\ell,0}\neq 0$ for some $\ell\in [s]$, we know that $B(x):=\sum_{\ell=1}^s a_{\ell,0} \cdot x^{\sigma^{\ell-1}}$ is a nonzero polynomial of degree at most $q^{s-1}$ over $\F_{q^m}$, and consequently, it has at most $q^{s-1}$ zeros in $\F_{q^m}$. Since the map $B(x)$ is $\F_q$-linear, we conclude that the kernel is an $\F_q$-linear subspace of dimension at most $s-1$, and so
 $\dim(\ker(M_0)) \leq s-1$.
\end{proof}

By \eqref{eq_rs_eqns} and Claim \ref{clm_rs_btt_ker} above, we have that all  polynomials $f=(f_0,f_1,\ldots,f_{k-1})$ that agree with  $\mathbf{y}$ on at least $t$ points are contained in an affine shift of the kernel of a $(k,r,m)$-$\btt$ matrix over $\F_q$ for $r \geq m-s+1$. In the next section, we prove that the kernel of a $(k,r,m)$-$\btt$ matrix is a $(k,m,m-r)$-$\btt$ subspace. Finally, noting that in our setting $m-r \leq s-1$, and that a basis for the kernel of $M$, as well as the desired affine shift (which is any valid solution to \eqref{eq_rs_eqns}), can be found in time $\poly(\log q,m,n)$,  
concludes the proof of Lemma \ref{lem_rs}.

\subsection{The kernel of a BTT matrix is a BTT subspace}\label{subsec_btt_ker_equiv}

In this section, we prove the following lemma.

 \begin{lem}\label{lem_btt_ker_equiv}
 Suppose that $M$ is a $(k,r,m)$-$\btt$ matrix over $\F_q$, where $r\leq m$. Then $\ker(M)$ is a $(k,m,m-r)$-$\btt$ subspace over $\F_q$.
\end{lem}

To prove the above lemma, let $V:= \ker(M)$, where
\[M=
\begin{pmatrix}
M_{1} & 0 & 0 & \cdots & 0  \\
M_{2} & M_{1} & 0 & \cdots & 0 \\
M_{3} & M_{2} & M_{1} & \cdots & 0 \\
\vdots  & \vdots  & \vdots & \ddots & \vdots  \\
M_{k} & M_{k-1} & M_{k-2} & \cdots & M_{1} 
\end{pmatrix}
\]
is a $(k,r,m)$-$\btt$ matrix.
 Our goal is to show that $V$ is a $(k,m,m-r)$-$\btt$ subspace, and for this we need to exhibit a $(k,m,m-r)$-$\btt$ matrix $\tilde M$ so that $V=\img(\tilde M)$. 

We start by introducing some notation. We write a vector $v\in\F_q^{km}$ as $v=(v_1,v_2,\dots,v_k)$ where $v_i\in \F_q^m$.
 For $i=0,1,\dots, k$, 
we let
$$V_i=\left\{ (v_1,v_2,\dots,v_k) \in V \mid v_1=v_2 = \cdots =v_{i}=0 \right\}.$$
In particular, we have $V_0=V$ and $V_k=\{0\}$.
Finally, define $\sigma: \F_q^{km}\to \F_q^{km}$ by
\[
\sigma(v_1,v_2,\dots,v_k)=(0, v_{1},v_{2},\dots,v_{k-1}).
\]

\begin{claim}\label{clm_btt_ker_contain}
For all $i=1,\ldots,k$,   $\sigma(V_{i-1}) \subseteq V_{i} \subseteq V_{i-1}$.
\end{claim}

\begin{proof}
The right-hand containment clearly holds by the definition of $V_i$.
To see that the left-hand containment holds, let $v=(v_1,\ldots,v_k) \in V_{i-1}$, and let $u=\sigma(v)=(0,v_1,\ldots,v_{k-1})$. Our goal is to show that $u\in V_{i}$.
First note that since $v \in V_{i-1}$, we have that $v_1=\cdots = v_{i-1}=0$, and so $u=(u_1,\ldots,u_k)$ satisfies that $u_1=\cdots =u_i=0$. Thus to show that $u \in V_{i}$, it
remains to show that $u \in V$, or equivalently that $M \cdot u =0$. 

To this end, note that by the structure of $M$,
$$M \cdot u = \begin{pmatrix}
M_{1} & 0 & 0 & \cdots & 0  \\
M_{2} & M_{1} & 0 & \cdots & 0 \\
M_{3} & M_{2} & M_{1} & \cdots & 0 \\
\vdots  & \vdots  & \vdots & \ddots & \vdots  \\
M_{k} & M_{k-1} & M_{k-2} & \cdots & M_{1} 
\end{pmatrix} \cdot \begin{pmatrix} 0 \\ v_1 \\ \vdots \\ v_{k-1} \end{pmatrix}  
= 
\begin{pmatrix}
M_{1} & 0 &  \cdots & 0  \\
M_{2} & M_{1} &  \cdots & 0 \\
\vdots  & \vdots  & \ddots & \vdots  \\
M_{k-1} &  M_{k-2} & \cdots & M_{1} 
\end{pmatrix} \cdot \begin{pmatrix}  v_1 \\ \vdots \\ v_{k-1} \end{pmatrix} =0,$$
where the last equality follows since $v \in V$, and so 
$$M \cdot v = \begin{pmatrix}
M_{1} & 0 & 0 & \cdots & 0  \\
M_{2} & M_{1} & 0 & \cdots & 0 \\
M_{3} & M_{2} & M_{1} & \cdots & 0 \\
\vdots  & \vdots  & \vdots & \ddots & \vdots  \\
M_{k} & M_{k-1} & M_{k-2} & \cdots & M_{1} 
\end{pmatrix} \cdot \begin{pmatrix}   v_1 \\ \vdots \\ v_{k} \end{pmatrix}=0.
$$
\end{proof}

We also note the following claim which follows by counting the number of linearly-independent constraints defining $V_i$.

\begin{claim}\label{clm_btt_ker_dim}
For all $i=0,1,\ldots,k$, $\dim(V_i) = (k-i)(m-r)$. 
In particular, $\dim(V_{i-1}) = \dim(V_{i})+(m-r)$ for all $i=1,\ldots,k$. 
\end{claim}

The above two claims imply the following.

\begin{claim}\label{clm_btt_ker_main}
The following holds  for all $i=1, \ldots, k-1$.
Suppose that $b^{(1)},\ldots,b^{(m-r)}$ are $m-r$ linearly independent vectors in $\F_q^{km}$ so that 
\begin{equation}\label{eq_btt_ker1}
V_{i-1} = V_{i} + \spn\{b^{(1)},\ldots,b^{(m-r)}\}.
\end{equation}
Then
\begin{equation}\label{eq_btt_ker2}
V_{i} = V_{i+1} + \spn\left\{ \sigma(b^{(1)}),\ldots,\sigma(b^{(m-r)})\right\}.
\end{equation}
\end{claim}

\begin{proof}
First note that by our assumption \eqref{eq_btt_ker1}, we have that $b^{(1)},\ldots,b^{(m-r)}$ are contained in $V_{i-1}$. By Claim \ref{clm_btt_ker_contain}, this implies in turn that  $\sigma(b^{(1)}),\ldots,\sigma(b^{(m-r)}) \in V_i$ and $V_{i+1} \subseteq V_i$, and consequently we have that the right-hand side of 
\eqref{eq_btt_ker2} is contained in the left-hand side.

To see the containment in the other direction, recall that by Claim \ref{clm_btt_ker_dim}, $\dim(V_i) - \dim(V_{i+1})=  m-r$, and so it suffices to show that there is no non-trivial linear combination of 
$\sigma(b^{(1)}),\ldots,\sigma(b^{(m-r)})$ that belongs to $V_{i+1}$. 
Suppose in contradiction that there exists a non-trivial linear combination $a := \alpha_1 \cdot \sigma(b^{(1)})+\cdots+\alpha_{m-r}\cdot \sigma(b^{(m-r)}) \in V_{i+1}$. By the definition of $V_{i+1}$, this implies in turn that $a_{i+1}=0$. But in this case, the non-trivial linear combination $a' :=\alpha_1 \cdot b^{(1)}+\cdots+\alpha_{m-r} \cdot b^{(m-r)}$ satisfies that $a'_i=0$. Consequently, we have that $a' \in V_{i}$, contradicting our assumption \eqref{eq_btt_ker1}.
\end{proof}

Now we  prove Lemma~\ref{lem_btt_ker_equiv} using the above claim.

\begin{proof}[Proof of  Lemma~\ref{lem_btt_ker_equiv}]
Recall that our goal is to exhibit a $(k,m,m-r)$-$\btt$ matrix $\tilde M$ so that $V=\img(\tilde M)$. We construct $\tilde M$ as follows. Since $\dim(V_0) = \dim(V_1)+(m-r)$, there exist $m-r$ linearly independent vectors $b^{(1)},\ldots,b^{(m-r)} \in \F_q^{km}$ so that 
$V_{0} = V_{1} + \spn\{b^{(1)},\ldots,b^{(m-r)}\}$.
For $i=1,\ldots,k$, let $M_i$ be a $(km) \times (m-r)$ matrix whose columns are $\sigma^{(i-1)}(b^{(1)}),\ldots,\sigma^{(i-1)}(b^{(m-r)})$. 
Let
\[
\tilde M=\begin{pmatrix} M_1  & M_2  & \cdots  & M_k\end{pmatrix}.
\] 
Then we clearly have that $\tilde M$ is a $(k,m,m-r)$-$\btt$ matrix. Moreover, by Claim \ref{clm_btt_ker_main}
we further have that $V_{i-1} = V_{i}+ \img(M_i)$ for all $i=1,\ldots k$, and so $V=V_0 = \img(\tilde M)$. This concludes the proof of Lemma \ref{lem_btt_ker_equiv}.
\end{proof}




 
\maketitle

\section{Preliminaries on function fields and algebraic-geometric codes}\label{sec:AG_pre}

We first give preliminaries and notations about function fields and algebraic-geometric codes. The reader may refer to, e.g., \cite{Sti09} for detailed background.

\paragraph{Function fields.}
Let $\F_q$ be a finite field.
An extension field $F$ of $\F_q$ is called a \emph{function field in one variable}  or simply a \emph{function field} over $\F_q$ if $F$ is a finite extension of $\F_q(x)$ for some element $x\in F$ that is transcendental over $\F_q$.
The \emph{field of constants} of   $F$  is the algebraic closure of $\F_q$ in $F$.

In the rest of this section,   let $F$ be a  function field  over $\F_q$  such that its field of constants is $\F_q$, i.e, the algebraic closure of $\F_q$ in $F$ is $\F_q$ itself.

\paragraph{Discrete valuations and places.}

A (normalized) \emph{discrete valuation} of $F$ is a map $v: F\to \Z\cup\{+\infty\}$ 
with the following properties:
\begin{itemize}
\item $v(a)=+\infty$ iff $a=0$.
\item $v(ab)=v(a)+v(b)$ for $a,b\in F$.
\item $v(a+b)\geq\min\{v(a), v(b)\}$ for $a,b\in F$.
\item $v(F^\times)= \Z$.
\end{itemize}
 
For a discrete valuation $v$ of $F$, we associate a pair $P=(\O_v, \m_v)$ where $\O_v$ is the ring $\{a\in F: v(a)\geq 0\}$
and $\m_v$ is the ideal  
$\{a\in \O_v: v(a)>0\}$ 
of $\O_v$.
 Call $P$ a \emph{place} of $F$.\footnote{It is common in the literature to define a place to be just the ideal $\m_v$ associated with a discrete valuation $v$ instead of $(\O_v, \m_v)$  (see, e.g., \cite{Sti09}). This is equivalent to our definition since $\O_v$ is determined by $\m_v$ via $\O_v=\{a\in F^\times: a^{-1}\not\in\m_v\}\cup\{0\}$.}
Denote by $\mathbb{P}(F)$ the set of all places of $F$, i.e.,
\[
\mathbb{P}(F):=\{(\O_v, \m_v): v~\text{is a discrete valuation of}~F\}.
\]

We may recover the discrete valuation $v$ from a place $P=(\O, \m)$ as follows. Let $v(0)=+\infty$. For $0\neq a\in \O$, $v(a)$ is the largest $k\in\N$ such that $a\in \m^k$, where we let $\m^0=\O$. For $a\in F^\times\setminus \O$, let $v(a)=-v(a^{-1})$. 
This gives a one-to-one correspondence between the set $\mathbb{P}(F)$ of all places of $F$  and the set of all discrete valuations of $F$.
For a place $P\in \mathbb{P}(F)$, denote by $v_P$ the discrete valuation corresponding to $P$.

Intuitively, $v_P(f)$ indicates the order of zeros or poles of a function $f\in F$ at the place $P$: If $v_P(f)\geq 0$, then $v_P(f)$ is the order of zeros of $f\in F$ at $P$. Otherwise $-v_P(f)$ is the order of poles of $f$ at $P$.

It can be shown that for a place $P=(\O, \m)$ of $F$, the quotient ring $\kappa_P:=\O/\m$ is a finite field extension of $\F_q$, called the  \emph{residue class field} or  \emph{residue field} of $P$.  If $[\kappa_P: \F_q]=1$, we say the place $P$ is \emph{$\F_q$-rational} or simply \emph{rational}. In this case, we identify $\F_q$ with $\kappa_P$  via the field isomorphism $\F_q\to \kappa_P$ sending $a\in\F_q$ to $a+\m$.

For $f\in \O$ and a rational place $P$ of $F$, define  
\[
f(P):=f+\m\in\kappa_P
\]
which we view as an element of $\F_q$ by identifying $\F_q$ with $\kappa_P$ as above.

\paragraph{Local  power series and  Laurent series expansion.}

Let $P=(\O, \m)$ be a rational place of $F$. An element $u\in\O$ is called a \emph{uniformizing parameter} or \emph{uniformizer} of $P$  if $v_P(u)=1$, or equivalently, $u$ generates the ideal $\m$. 

Fix $u\in\O$ to be a uniformizer of $P$. We may write any $f\in \O$ as a \emph{power series} in $u$ over $\F_q$
\[
f = c_0+c_1 u+ c_2 u^2+\cdots 
\]
where the coefficients $c_i\in\F_q$ may be found as follows:  Let $f_0=f$. For $i=0,1,2,\dots$, let $c_i=f_i(P)$ and let $f_{i+1}=(f_i-c_i)/u\in\O$.

A \emph{Laurent series} is a generalization of a power series, where we allow finitely many terms of negative degree.
Generalizing the above representation by power series, we may write any element of $F$ as a Laurent series in $u$ over $\F_q$.
Namely, for $f\in F^\times$, let $e=v_P(f)$ and $f^*=f/u^e\in\O$. Suppose $f^*=c_0+c_1 u+ c_2 u^2+\cdots$.
Then
\[
f = u^e f^*=c_0u^{e}+c_1 u^{e+1}+ c_2 u^{e+2}+\cdots.
\]

 Thus, for a rational place $P=(\O, \m)$ and a uniformizer $u$ of $P$, we have a local expansion  of  every element of $\O$ or $F$ as a power series or a Laurent series in $u$ over $\F_q$, respectively.

%

\paragraph{Divisors.} 

A \emph{divisor} of $F$ is a formal sum $\sum_P n_P P$ of finitely many places $P\in \mathbb{P}(F)$, where $n_P\in \Z$.
The set of all divisors of $F$ forms an abelian group $\Div(F)$, called the \emph{divisor group} of $F$. 

The \emph{degree} of a divisor $D=\sum_P n_P P$ is $\deg(D):=\sum_P n_P [\kappa_P: \F_q]$.
The \emph{support} of $D$, denoted by $\supp(D)$, is the set of places $P$ for which $n_P\neq 0$. 
If $n_P\geq 0$ for all $P\in\supp(D)$, we write $D\geq 0$ and call $D$ an \emph{effective divisor}.
Note that $D\geq 0$ implies $\deg(D)\geq 0$.
Let $\Div_0(F):=\{D\in\Div(F): \deg(D)=0\}$, which is a subgroup of $\Div(F)$.

Let $f\in F^\times$.  It can be shown that $v_P(f)=0$ holds for all but finitely many places $P\in\mathbb{P}(F)$. So $\div(f):=\sum_{P\in\mathbb{P}(F)} v_P(f) P$ is a well-defined divisor. Divisors of the form $\div(f)$ are called \emph{principal divisors} of $F$. The degree of a principal divisor is always zero, i.e., $\div(f)\in \Div_0(F)$ for $f\in F^\times$.

\paragraph{Riemann--Roch spaces.} For a divisor $D$ of $F$, the \emph{Riemann--Roch space} associated with $D$ is
\[
L(D):=\{f\in F^\times: \div(f)+D\geq 0\}\cup \{0\}
\]
which is a finite-dimensional  vector space over $\F_q$. Let $\ell(D):=\dim_{\F_q} L(D)$. 

By definition, for $D=\sum_P n_P P$, the condition $\div(f)+D\geq 0$ is equivalent to $v_P(f)\geq -n_P$ for $P\in \mathbb{P}(F)$.
So  $L(D)$ is the space of functions in $F$ whose prescribed zeros and allowed poles are specified by $D$: At a place $P$, if $n_P<0$, then any $f\in L(D)$ must have a zero of order at least $-n_P$ at $P$. On the other hand, if $n_P\geq 0$, then $f\in L(D)$ is allowed to have a pole of order at most $n_P$ at $P$. 

Note that if $L(D)$ contains a nonzero element $f$, then $\div(f)+D\geq 0$, which implies 
\[
\deg(D)=\deg(\div(f))+\deg(D)=\deg(\div(f)+D)\geq 0.
\] 
So for any divisor $D$ with $\deg(D)<0$, we have $L(D)=\{0\}$ and $\ell(D)=0$.

\paragraph{The Riemann--Roch theorem.}

The \emph{Riemann--Roch theorem} states that
\[
\ell(D)-\ell(K-D)=\deg(D)-g+1
\]
holds for any divisor $D$ of $F$, where $K$ is a certain divisor of $F$  called a \emph{canonical divisor},  and $g$ is a nonnegative integer depending only on $F$ called the \emph{genus} of $F$.

In fact, we only need  the following corollary of the Riemann--Roch theorem.
\begin{thm}[Riemann's inequality]\label{thm_riemann}
 $\ell(D)\geq \deg(D)-g+1$.
 \end{thm}

\paragraph{Algebraic-geometric codes.}

Let $D$ be a divisor of $F$ and let $S=\{P_1,P_2,\dots,P_n\}$ be a set of $n$ distinct rational places of $F$ such that $\supp(D)\cap S=\emptyset$.
Define the \emph{algebraic-geometric ($\ag$) code}
\begin{equation}\label{eq_AG}
C(S, D):=\{(f(P_1), f(P_2),\dots, f(P_n)): f\in L(D)\}\subseteq \F_q^n,
\end{equation}
which is an $\F_q$-linear code of block length $n$. 

Let $D_S=\sum_{P\in S} P\in\Div(F)$. We have the following theorem.
\begin{thm}[{\cite[Theorem~2.2.2]{Sti09}}]\label{thm_generalAG}
The dimension of $C(S,D)$ is $\ell(D)-\ell(D-D_S)$, which equals $\ell(D)$ if $\deg(D)<\deg(D_S)=n$.
The minimum distance of $C(S,D)$ is at least $n-\deg(D)$.
\end{thm}

 \subsection{Constant field extensions of function fields}\label{sec_constantext}
 
 Let $\F_{q^m}/\F_q$ be a finite field extension of degree $m\in\N^+$. Denote by $F^{(m)}$ the compositum $F\F_{q^m}$ of $F$ and $\F_{q^m}$.
 Then $F^{(m)}$ is a function field over $\F_{q^m}$.
 Recall that we assume the field of constants of $F$ is $\F_q$. 
 This implies that the field of constants of $F^{(m)}$ is $\F_{q^m}$ \cite[Proposition~3.6.1]{Sti09}. 
 
 \paragraph{Places and divisors of $F^{(m)}$.}
 
Let $P=(\O, \m)$ be a  \emph{rational} place of $F$ and  $v_P$ the corresponding discrete valuation of $F$.
It can be shown that there exists a \emph{unique} discrete valuation $v_P'$ of $F^{(m)}$ that extends $v_P$. 
 We denote the corresponding place of $F^{(m)}$ by $P^{(m)}=(\O^{(m)}, \m^{(m)})$, which is an $\F_{q^m}$-rational place.
 As $v_P'$ extends $v_P$, we have $\O\subseteq \O^{(m)}$ and $\m\subseteq \m^{(m)}$.
 So $f(P^{(m)})=f(P)$ for any $f\in \O\subseteq \O^{(m)}$.
And a uniformizer $u\in \m$ of $P$ is also a uniformizer of $P^{(m)}$.
 
Let $D=\sum_{P} n_P P$ be a divisor of $F$ such that every $P\in\supp(D)$ is rational. 
Then we define 
\[
D^{(m)}:=\sum_{P\in \supp(D)} n_P P^{(m)},
\]
which is a divisor of $F^{(m)}$. 
The Riemann-Roch space $L(D^{(m)})$ and its dimension $\ell(D^{(m)})$ are defined as before, except that the base field is changed to $\F_{q^m}$.
 That is, 
  \[
  L(D^{(m)})=\{f\in (F^{(m)})^\times: \div(f)+D^{(m)}\geq 0\}\cup \{0\} 
  \]
  and $\ell(D^{(m)})=\dim_{\F_{q^m}} L(D^{(m)})$.

The following lemma is a special case of \cite[Theorem~3.6.3\,(d)]{Sti09}.
  
\begin{lem}\label{lem_dimbc}
Let $D$ be as above. If $f_1, \dots, f_k\in F$ form a basis of $L(D)$ over $\F_q$, then they form a basis of $L(D^{(m)})$ over $\F_{q^m}$. In particular, $\ell(D^{(m)})=\ell(D)$.
\end{lem} 

\begin{rem}
We only need the definition of $D^{(m)}$ for the special case that every $P\in\supp(D)$ is rational, but it can also be defined for a general divisor $D$. See \cite[Definition~3.1.8]{Sti09} for the general definition, where it is called the \emph{conorm} of $D$ and denoted by $\mathrm{Con}_{F^{(m)}/F}(D)$.

We also note that $P^{(m)}$ and $D^{(m)}$ above are simply denoted by $P$ and $D$ respectively  in \cite{GX13} by a slight abuse of notation. 
\end{rem}
 
 \paragraph{The Frobenius automorphism.}

  Let $\sigma$ be the Frobenius automorphism $a\mapsto a^q$ of $\F_{q^m}$ over $\F_q$.
  As $\F_q$ is the field of constants of $F$, we have $F\cap \F_{q^m}=\F_q$.
  This implies that $\gal(F^{(m)}/F)$ is isomorphic to  $\gal(\F_{q^m}/\F_q)$ via the restriction map $\tau\mapsto \tau|_{\F_{q^m}}$ (see \cite[Theorem~1.12]{Lan02}).
  So $\sigma\in \gal(\F_{q^m}/\F_q)$ uniquely extends to an automorphism of $F^{(m)}$ that fixes $F$, which we also call $\sigma$ by an abuse of notation.
  
  As $\sigma$ is an automorphism of $F^{(m)}$, it permutes the places of  $F^{(m)}$.
Let $P$ be any rational place of $F$. As $\sigma$  fixes $F$, it also fixes $P$. So  $\sigma$ also fixes $P^{(m)}$.
  

 \subsection{The Garcia--Stichtenoth tower}\label{sec_gstower}
 
 We need the following tower of function fields introduced by Garcia and Stichtenoth in \cite{GS96}.
 
\begin{defi}[Garcia--Stichtenoth tower \cite{GS96}]
 Let $r>1$ be a prime power and $q=r^2$.
 For $i=1,2,\dots$, let $K_i=\F_q(x_1,x_2,\dots,x_i)$, where $x_1$ is transcendental over $\F_q$ and $x_i$ satisfies the following recursive equation for $i>1$.
\[
x_{i}^r+x_{i}=\frac{x_{i-1}^r}{x_{i-1}^{r-1}+1}.
\]
The \emph{Garcia--Stichtenoth tower} over $\F_q$ is the infinite tower of function fields $K_1\subseteq K_2\subseteq \cdots$. 
\end{defi}

For each $e\in\N^+$, we have $[K_e: K_1]=r^{e-1}$ and the field of constants of $K_e$ is $\F_q$.

\paragraph{Rational places.}

Let $e\in\N^+$.
The field $K_e$ has at least $r^e(r-1)+1$ rational places. One of them is the place $P_\infty$ ``at the infinity,'' which is totally ramified over $K_1$ and is the unique pole of $x_1$, i.e., $v_{P_\infty}(x_1)=-[K_e: K_1]=-r^{e-1}$. More generally, we have $v_{P_\infty}(x_i)=-r^{e-i}$ for $i\in [e]$. In particular, $v_{P_\infty}(x_e)=-1$ and hence $x_e^{-1}$ is a uniformizer of $P_\infty$.

In addition, define $S_e$ to be the set of all tuples $\alpha=(\alpha_1,\alpha_2,\dots,\alpha_e)\in\F_q^e$ such that $\alpha_1^r+\alpha_1\neq 0$ and $\alpha_{i}^r+\alpha_{i}=\frac{\alpha_{i-1}^r}{\alpha_{i-1}^{r-1}+1}$ for $i=2, 3,\dots, e$.
For each $\alpha\in S_e$, there exists a corresponding rational place $P_\alpha$ of $K_e$. 
It is the unique rational place $P$ satisfying $v_{P}(x_i)\geq 0$ and $x_i(P)=\alpha_i$ for $i=1,2,\dots,e$.
 
 There are precisely $r^e(r-1)$ elements in $S_e$, corresponding to $r^e(r-1)$ rational places $P_\alpha$ of $K_e$.  

%

\paragraph{Genus.} 

For $e\in \N^+$, the genus $g(K_e)$ of $K_e$ is given by
\[
g(K_e)=\begin{cases}
(r^{e/2}-1)^2 & e \text{ is even,}\\
(r^{(e-1)/2}-1)(r^{(e+1)/2}-1) & e \text{ is odd.}
\end{cases}
\]
In particular, we have $g(K_e)\leq r^e$.

\paragraph{Explicitness.}

To construct $\ag$ codes using the Garcia--Stichtenoth tower, we need to construct bases for Riemann--Roch spaces  of   $K_e$.
An efficient algorithm of computing such bases  was given in \cite{SAKSD01} for one-point divisors $k P_\infty$.

\begin{thm}[\cite{SAKSD01}]\label{thm_explicitbase}
For $k\in \N$ and $e\in \N^+$, a basis $B$ of the Riemann--Roch space $L(k P_\infty)$ of $K_e$ over $\F_q=\F_{r^2} $ can be found in time  $\poly(k, r^e)$.
Moreover, given $\alpha\in S_e$ and $f\in L(k P_\infty)$ (represented  in the basis $B$), the evaluation $f(P_\alpha)$  can also be found in time $\poly(k, r^e)$. 
\end{thm}

In addition, it was shown in \cite{GX12} that the Laurent series expansion of $f\in L(k P_\infty)$ at the place $P_\infty$ in the uniformizer $x_e^{-1}$ can be computed efficiently.

\begin{lem}[\cite{GX12}]\label{lem_laurent}
Given $f\in L(k P_\infty)$ and $N\in\N^+$, the first $N$ coefficients $c_0, c_1, \dots, c_{N-1}\in \F_q$ of the Laurent series expansion
\[
f=c_0 T^{-k} + c_1 T^{-k+1} + c_2 T^{-k+2} + \cdots
\]
at the place $P_\infty$ in the uniformizer $T=x_e^{-1}$ can be found in time $\poly(k, r^e, N)$.
\end{lem}

 \section{Algebraic-geometric codes with subfield evaluation points}\label{sec:AG}
 
In this section, we present the proof of Theorem \ref{thm:intro_ag_btt}, which is restated below.

\agbtt*

 It is based on $\ag$ codes with subfield evaluation points and closely follows \cite{GX13}.
Specifically, we use   the Garcia--Stichtenoth tower of function fields discussed in Subsection~\ref{sec_gstower}.
We note that this framework is generic and can be adapted to work for other families of function fields as well. For more details, see Remark~\ref{rem_generic} at the end of this section.


Our construction is given below as Definition~\ref{defn:AG-subfield}. It uses the constant field extension $K_e^{(m)}$ of $K_e$, where $K_e$ is the $e$th field in the  Garcia--Stichtenoth tower over $\F_q=\F_{r^2}$.
Recall that $K_e$  has a rational place $P_\infty$ and  $r^e(r-1)$ rational places $P_\alpha$ for $\alpha\in S_e$. 
And for each rational place $P$ of $K_e$, there is a corresponding $\F_{q^m}$-rational place $P^{(m)}$ of $K_e^{(m)}$.
 
 \begin{defi}[$\ag$ codes with subfield evaluation points from the Garcia--Stichtenoth tower]\label{defn:AG-subfield}
  Let $r>1$ be a prime power and $q=r^2$.
  Let $n, k, m, e\in\N^+$ be such that $k\leq n$ and $n\leq r^e(r-1)$.  
  Let $\alpha_1,\dots, \alpha_n\in S_e\subseteq\F_q^e$ be distinct and let $P_i=P_{\alpha_i}^{(m)}$ for $i\in [n]$.
The code \emph{$\mathsf{GS}_{q, m, e}(n,k)$ over $\F_{q^m}$ with evaluation points $P_1, \dots, P_n$} maps
 $f\in L((k-1)P_\infty^{(m)})\subseteq K_e^{(m)}$  to the codeword
$C_f:=(f(P_1), f(P_2), \dots, f(P_n)) \in (\F_{q^m})^n$.
\end{defi}

\paragraph{Explicitness.} By Theorem~\ref{thm_explicitbase}, a basis $B$ of $L((k-1)P_\infty)$ over $\F_q$ can be computed in time $\poly(k, r^e)=\poly(r^e)$. Suppose $B=\{\beta_1, \beta_2, \dots, \beta_b\}$, where $b=\ell((k-1)P_\infty)$. By Lemma~\ref{lem_dimbc}, $B$ is also a basis of $L((k-1)P_\infty^{(m)})$ over $\F_{q^m}$. So we may write $f\in L((k-1)P_\infty^{(m)})$ uniquely as a linear combination of $\beta_i$ over $\F_{q^m}$:
\begin{equation}\label{eq_functionrep}
f=\sum_{i=1}^b c_i \beta_i, \quad \text{where}~c_i\in \F_{q^m}.
\end{equation}
We represent $f$ by the coefficients $c_1,\dots, c_b$ in the basis $B$.
Note 
\[
f(P^{(m)})=\left(\sum_{i=1}^b c_i \beta_i\right)(P^{(m)})=\sum_{i=1}^b c_i \beta_i(P)
\]
 for any rational place $P$ of $K_e$.
So by Theorem~\ref{thm_explicitbase}, the encoding map  $\mathsf{Enc}: L((k-1)P_\infty^{(m)}) \to (\F_{q^m})^n$ sending $f$ to $C_f=(f(P_1), f(P_2), \dots, f(P_n))$ can be computed in time $\poly(k, r^e, n, m\log q)=\poly(r^e, m)$.

 
\paragraph{Rate and minimum distance.}
 
 Denote by $g$ the genus of $K_e$. The following theorem bounds the rate and the minimum distance of the code $\mathsf{GS}_{q, m, e}(n,k)$.
 
 \begin{thm}\label{thm_gs_rate}
 $\mathsf{GS}_{q, m, e}(n,k)$ is a linear code over the alphabet $\F_{q^m}$  with block length $n$.
 Its rate is at least $(k-g)/n$ and its minimum distance is  at least $n-k+1$. 
\end{thm}
\begin{proof}
Let $S=\{P_1, \dots, P_n\}$ and $D=(k-1)P_\infty^{(m)}$. 
Then $\mathsf{GS}_{q, m, e}(n,k)$ is simply the linear code $C(S, D)$ defined in \eqref{eq_AG} with the base field replaced by $\F_{q^m}$.
By Theorem~\ref{thm_generalAG}, its dimension is $\ell((k-1)P_\infty^{(m)})$ and  its  minimum distance is at least $n-\deg((k-1)P_\infty^{(m)})=n-k+1$. 
By Lemma~\ref{lem_dimbc} and Riemann's inequality (Theorem~\ref{thm_riemann}), we have
$\ell((k-1)P_\infty^{(m)})=\ell((k-1)P_\infty)  \geq  k-g$.
So the rate of $\mathsf{GS}_{q, m, e}(n,k)$ is at least $(k-g)/n$.
 \end{proof}

\paragraph{The embedding $\phi$.} 
To list-decode the code  $\mathsf{GS}_{q, m, e}(n,k)$, we need an embedding (i.e., injective  linear  map)
\[
\phi: L((k-1)P_\infty^{(m)})\to \F_{q^m}^k.
\] 
It is defined to be the $\F_{q^m}$-linear map that outputs the first $k$ coefficients of the Laurent series expansion at the place $P_\infty^{(m)}$  in the uniformizer $T:=x_e^{-1}$.
That is, if the Laurent series expansion of $f\in L((k-1)P_\infty^{(m)})$ at $P_\infty^{(m)}$ in $T$ is
\[
f=f_0 T^{-(k-1)} + f_1 T^{-(k-1)+1} +  f_2 T^{-(k-1)+2} + \cdots,
\]
with the coefficients $f_i\in\F_{q^m}$, then $\phi(f)=(f_0, f_1, \dots, f_{k-1})$.

The kernel of $\phi$ is $L(-P_\infty^{(m)})=\{0\}$. So $\phi$ is indeed an embedding.
Representing a function $f\in L((k-1)P_\infty^{(m)})$ in the form \eqref{eq_functionrep}, we can compute $\phi(f)$ from $f$ in time $\poly(r^e, m)$  by Lemma~\ref{lem_laurent}.

\paragraph{List decoding.}
 
Next, we show that for properly chosen parameters, the code $\mathsf{GS}_{q, m, e}(n,k)$ is  list decodable up to the relative distance $1-R-\epsilon$ and that the image of the output list under the embedding $\phi$ is contained in an affine shift of a low-dimensional BTT subspace.

\begin{thm}\label{thm_ag}
  Let $\epsilon >0$ and $R \in (0,1-\epsilon)$.
  Let $e\in \N^+$ be a growing parameter.
  Let $r\geq 4/\epsilon+1$ be a prime power  and $q=r^2$.
  Choose $n,m,k\in\N^+$ such that $m \geq 4/\epsilon^2$, $4r^e/\epsilon \leq n\leq  (r-1)r^e$ and $k= \lceil R n + r^e\rceil\leq n$.
  Then $\mathsf{GS}_{q, m, e}(n,k)$ has rate at least $R$. And it can be list decoded from up to a $(1-R-\epsilon)$-fraction of errors with a list  of candidate messages whose images under $\phi$ (viewed as length $k m$ vectors over $\F_q$)  are contained in an affine shift of a $(k,m,\epsilon m)$-BTT subspace $V$ over $\F_q$.
Moreover, a basis for $V$ and the affine shift can be found in time $\poly(n, m)$ given the received word.
\end{thm}

The above theorem is a consequence of the following lemma.

\begin{lem}\label{lem_ag}
  Let $n, k,m,e,r,q\in\N^+$ and $\mathsf{GS}_{q, m, e}(n,k)$ be as in Definition~\ref{defn:AG-subfield}. Let $s\in [m]$ and $t, d\in\N^+$ be parameters, 
 satisfying that
 \begin{equation}\label{eq_ag_cond1}
(s+1)(d-g+1)+k-1>n
\end{equation}
and
\begin{equation}\label{eq_ag_cond2}
 t>d+k-1
\end{equation}
where $g$ is the genus of $K_e$.
  Then $\mathsf{GS}_{q, m, e}(n,k)$ can be list decoded from agreement at least $t$ with a list of candidate messages whose images under $\phi$ (viewed as length $k m$ vectors over $\F_q$) are contained in an affine shift of a $(k,m,s-1)$-BTT subspace $V$ over $\F_q$. Moreover, a basis for $V$ 
   and the affine shift
  can be found in time $\poly(r^e, m)$.
\end{lem}

Before we prove the above lemma, we show how it implies Theorem \ref{thm_ag}.

 \begin{proof}[Proof of Theorem \ref{thm_ag}]
 We know $g\leq r^e$. So the rate of  $\mathsf{GS}_{q, m, e}(n,k)$ is at least $(k-g)/n\geq (k-r^e)/n\geq R$ by Theorem~\ref{thm_gs_rate}.

Let $\epsilon'=\epsilon/4$. By assumption, we have $r^e/n\leq \epsilon'$.
Let $s= \frac {1}{\epsilon'} +1$, let $d=  \frac{n-k+2}{s+1}+g-1$ so that \eqref{eq_ag_cond1} is satisfied, and 
let $t=d+k$ so that \eqref{eq_ag_cond2} is satisfied. Then with this setting of parameters, we know from Lemma~\ref{lem_ag} that $\mathsf{GS}_{q, m, e}(n,k)$ can be list decoded from agreement  $t$, or equivalently, from up to
\begin{align*}
n-t &= n-d-k = n- k +1 -g- \frac{n-k+2}{s+1} \geq \left(1-\frac{1}{s+1}\right)(n- k+1)  -g- \frac{1}{s+1} \\
&\geq (1-\epsilon')(n-Rn-r^e) - r^e - \epsilon'\\
&\geq (1-R-4\epsilon') n\\
&= (1-R-\epsilon) n
\end{align*}
errors. 
Moreover, as $m \geq 4/\epsilon^2=1/(\epsilon\cdot \epsilon')$, we have that $V$ is a $(k,m,s-1)$-BTT subspace for $s-1 = \frac {1}{\epsilon'} \leq \epsilon m$.  And a basis for $V$ as well as the affine shift can be found in time $\poly(r^e, m)=\poly(n, m)$ by  Lemma~\ref{lem_ag}.
\end{proof}

Theorem~\ref{thm:intro_ag_btt} follows easily from Theorem~\ref{thm_ag}.

\begin{proof}[Proof of Theorem~\ref{thm:intro_ag_btt}]
 Fix a prime power $r=O(1/\epsilon)$ such that $r\geq 4/\epsilon+1$. Let $q=r^2$ and  $m=\lceil 4/\epsilon^2\rceil$.
 Choose the family of codes to be
 \[
 \{  \mathsf{GS}_{q, m, e}(n,k): e\in \N^+,    4r^e/\epsilon \leq n\leq  (r-1)r^e, k= \lceil R n + r^e\rceil \}.
 \]
 Then Theorem~\ref{thm:intro_ag_btt} follows from Theorem~\ref{thm_ag}. 
\end{proof}

So it remains to prove Lemma~\ref{lem_ag}. We prove this lemma in the next two subsections.

 \subsection{The polynomial $Q$}\label{subsec_ag_Q}

In what follows, let $\sigma$ be the Frobenius automorphism $a\mapsto a^q$ of $\F_{q^m}$ over $\F_q$. 
It uniquely extends to an automorphism of $K_e^{(m)}$ that fixes $K_e$, which we also call $\sigma$ by an abuse of notation.
The automorphism $\sigma$ fixes $P^{(m)}$ for any rational place $P$ of $K_e$.
For $f\in K_e^{(m)}$, denote by $f^\sigma$ the element $\sigma(f)$.


Suppose that $\mathbf{y}=(y_1,y_2,\dots,y_n)\in (\F_{q^m})^n$ is a received word. 
We let $Q$ be a nonzero multivariate polynomial in
$K_e^{(m)}[Y_1, Y_2, \dots, Y_s]$ of the form
\[
Q=A_0 + A_1 Y_1 + A_2 Y_2 + \dots + A_s Y_s
\]
where $A_0, A_1, \dots, A_s\in  K_e^{(m)}$, $A_0\in L((d+k-1)P_\infty^{(m)})$, and $A_i\in  L(d P_\infty^{(m)})$ for $i=1, 2, \dots, s$.
We also require the coefficients $A_i$ to satisfy the constraint
\begin{equation}\label{eq_ag_constraint}
A_0(P_i) + A_1(P_i) y_i + A_2(P_i) y_i^\sigma + \dots + A_s(P_i) y_i^{\sigma^{s-1}}=0
\end{equation}
for all $i = 1,\ldots,n$, where $P_1,\dots, P_n$ are the evaluation points.

We first claim that such a nonzero polynomial $Q$ exists and can be computed efficiently. To see this, write $A_0$ as a vector over $\F_{q^m}$  with $\ell((d+k-1)P_\infty^{(m)})$ coordinates, and write $A_i$ as a vector over $\F_{q^m}$  with $\ell(d P_\infty^{(m)})$ coordinates for $i=1,\dots, n$.
Think of the coordinates of these vectors as unknowns. 
This gives
\begin{align*}
\ell((d+k-1)P_\infty^{(m)}) + s\cdot \ell(d P_\infty^{(m)})
&\geq (d+k-1)-g+1 + s (d-g+1)\\
&= (s+1)(d-g+1)+(k-1)
\end{align*}
 unknowns in total, where the first inequality above follows from Riemann's inequality (Theorem~\ref{thm_riemann}).
On the other hand, \eqref{eq_ag_constraint} gives $n$ homogeneous linear constraints in these unknowns over $\F_{q^m}$.
 By \eqref{eq_ag_cond1}, the number of unknowns is greater than the number of linear constraints which guarantees the existence of a nonzero solution $Q$. Moreover, we can  find $Q$ in time $\poly(r^e, m)$
 by constructing and then solving the system of linear equations represented by \eqref{eq_ag_constraint}.
 (Note $d$ is polynomial in $r^e$ since  $d<t$ by \eqref{eq_ag_cond2} and the agreement $t$ is bounded by $n\leq (r-1)r^e$.)

Next, we show that $Q$ gives a functional equation that any $f$ that has sufficiently large agreement with the received word $\mathbf{y}$ needs to satisfy.

\begin{claim}\label{clm_ag_funceqn}
Let $f\in  L((k-1)P_\infty^{(m)})$. Suppose $\mathbf{y}$ agrees with the codeword $C_f=(f(P_1), f(P_2), \dots, f(P_n))$ in at least $t$ coordinates.
Then $f$ satisfies the functional equation
\begin{equation}\label{eq_ag_funceqn}
Q(f, f^{\sigma},\dots, f^{\sigma^{s-1}})=A_0+A_1 f+ A_2 f^{\sigma}+\dots+A_s f^{\sigma^{s-1}}=0.
\end{equation}
\end{claim}

\begin{proof}
Define
\[
Q^*=A_0+A_1 f+ A_2 f^{\sigma}+\dots+A_s f^{\sigma^{s-1}}\in K_e^{(m)}.
\]
We want to prove that $Q^*=0$.
As $f\in  L((k-1)P_\infty^{(m)})$,  $A_0\in L((d+k-1)P_\infty^{(m)})$, $A_i\in  L(d P_\infty^{(m)})$ for $i=1, 2, \dots, s$, and $\sigma$ fixes $P_\infty^{(m)}$, we know $Q^*\in L((d+k-1)P_\infty^{(m)})$.

Suppose that $\mathbf{y}$ agrees with $C_f$ in the $i$-th symbol for some $i\in [n]$,   i.e., $y_i=f(P_i)$.
By \eqref{eq_ag_constraint}, we have
\begin{align*}
0&=A_0(P_i) + A_1(P_i) y_i + A_2(P_i) y_i^\sigma + \dots + A_s(P_i) y_i^{\sigma^{s-1}}\\
&=A_0(P_i) + A_1(P_i) f(P_i) + A_2(P_i) (f(P_i))^\sigma + \dots + A_s(P_i) (f(P_i))^{\sigma^{s-1}}\\
&=A_0(P_i) + A_1(P_i) f(P_i) + A_2(P_i) f^\sigma(P_i) + \dots + A_s(P_i) f^{\sigma^{s-1}}(P_i)\\
&=(A_0+A_1 f+ A_2 f^{\sigma}+\dots+A_s f^{\sigma^{s-1}})(P_i)\\
&=Q^*(P_i).
\end{align*}
The third equality uses the fact that $(f(P_i))^\sigma=f^\sigma(P_i)$, which holds since $P_i=P_{\alpha_i}^{(m)}$ is fixed by $\sigma$.

As  $\mathbf{y}$ and $C_f$ agree in at least $t$ symbols, the above argument shows that there exist $i_1,\dots, i_t\in [n]$ such that $Q^*$ vanishes at $P_{i_1},\dots, P_{i_t}$.
Let $D=\sum_{j=1}^t P_{i_j}$.
Then $Q^*\in L((d+k-1)P_\infty^{(m)}-D)$.
On the other hand, the degree of the divisor $(d+k-1)P_\infty^{(m)}-D$ is $d+k-1-t$, which is less than zero by \eqref{eq_ag_cond2}. So $L((d+k-1)P_\infty^{(m)}-D)=\{0\}$. This implies $Q^*=0$.
\end{proof}

\subsection{The BTT subspace $V$}\label{subsec_ag_V}
Next, we show that the functional equation (\ref{eq_ag_funceqn}), given by Claim \ref{clm_ag_funceqn} above, implies that the image of the list of candidate messages under the embedding $\phi$ is contained in an affine shift of a low-dimensional BTT subspace. We start by expanding the functional equation (\ref{eq_ag_funceqn}) in terms of the coefficients of the polynomials $f$ and $A_0,A_1,\ldots,A_s$.

Suppose that $f\in  L((k-1)P_\infty^{(m)})$ agrees with 
$\mathbf{y}$ in at least $t$ coordinates.
Consider the Laurent series expansion of $f$ at $P_\infty^{(m)}$ in the uniformizer $T=x_e^{-1}$:
\[
f=\sum_{i=0}^{\infty} f_i T^{-(k-1)+i}
\] 
where the coefficients $f_i$ are in $\F_{q^m}$. 
As $T\in K_e$ is fixed by $\sigma$, we have
$f^{\sigma^j}=\sum_{i=0}^{\infty} f_i^{\sigma^j} T^{-(k-1)+i}$ for any integer $j$.
By definition, $\phi(f)=(f_0, f_1, \dots, f_{k-1})$.

Similarly, expand $A_0\in L((d+k-1)P_\infty^{(m)})$ and $A_1,\dots, A_s\in  L(d P_\infty^{(m)})$ 
as Laurent series at $P_\infty^{(m)}$ in the uniformizer $T$:
\begin{equation}\label{eq_laurent}
A_0=\sum_{i=0}^{\infty} a_{0,i} T^{-(d+k-1)+ i} \quad\text{and}\quad A_\ell=\sum_{i=0}^{\infty} a_{\ell,i} T^{-d+i}, \quad \ell=1,\dots,s
\end{equation}
where the coefficients $a_{\ell,i}$ are in $\F_{q^m}$ for $\ell=0,1,\dots,s$ and $i\in\N$.
Choose the largest integer $u\geq 0$ such that 
there exists $\ell_0\in \{0,1,\dots,s\}$ satisfying $a_{\ell_0,u}\neq 0$.
By \eqref{eq_ag_funceqn}, we may assume $\ell_0 \in [s]$.
 (Otherwise, we have $a_{0,u}\neq 0$ and $a_{1,u}=\dots=a_{s,u}=0$. Then the LHS of \eqref{eq_ag_funceqn}, which we denote by $Q^*$, satisfies $v_{P_\infty^{(m)}}(Q^*)=-(d-k+1)+u<+\infty$, contradicting \eqref{eq_ag_funceqn}.)
 Then we have $0\neq A_{\ell_0}\in L((d-u)P_\infty^{(m)})$, which implies $u\leq d$.
 

Let $\hat{a}_{\ell, i}=a_{\ell, i+u}$ for $\ell=0,1,\dots, s$ and $i\in\N$. So $\hat{a}_{\ell_0, 0}=a_{\ell_0, u}\neq 0$. We may rewrite \eqref{eq_laurent} as
\[
A_0=\sum_{i=0}^{\infty} \hat{a}_{0,i} T^{-(d+k-1)+u+ i} \quad\text{and}\quad A_\ell=\sum_{i=0}^{\infty} \hat{a}_{\ell,i} T^{-d+u+i}, \quad \ell=1,\dots,s.
\] 

With the notations above, \eqref{eq_ag_funceqn} becomes
\begin{align*}
0&=\sum_{i=0}^{\infty} \hat{a}_{0,i} T^{-(d+k-1)+u+i} +\sum_{\ell=1}^s \left( \sum_{i=0}^{\infty} \hat{a}_{\ell,i} T^{-d+u+i} \right) \left(  \sum_{i=0}^{\infty} f_i^{\sigma^{\ell-1}} T^{-(k-1)+i}\right) \\
&=  \sum_{i=0}^{\infty}\left(\hat{a}_{0,i}+\sum_{\ell=1}^s \sum_{j=0}^{i} \hat{a}_{\ell,i-j} f_{j}^{\sigma^{\ell-1}} \right) T^{-(d+k-1)+u+i}.
\end{align*}
So we obtain the equations
\begin{equation}\label{eq_ag_eqns}
 \sum_{j=0}^{i} \sum_{\ell=1}^s \hat{a}_{\ell,i-j} f_{j}^{\sigma^{\ell-1}}  = - \hat{a}_{0,i}, \qquad i=0,1, \dots, k-1
\end{equation}
where $\hat{a}_{\ell,0}\neq 0$ for some $\ell\in [s]$.

By Claim~\ref{clm_rs_btt_ker}, the solution set of all $\phi(f)=(f_0,f_1,\ldots,f_{k-1})$ satisfying \eqref{eq_ag_eqns} is contained in an affine shift of the kernel of a $(k,r,m)$-BTT matrix $M$ over $\F_q$  for some $r \geq m-s+1$, and $M$ can be constructed in time $\poly(\log q,m,n)$ given the coefficients $\hat{a}_{\ell, i}$. 
By Lemma~\ref{lem_btt_ker_equiv}, the kernel of $M$ is a $(k,m,m-r)$-BTT subspace. It is a subspace of a $(k,m,s-1)$-BTT subspace since $m-r\leq s-1$.

To compute $\hat{a}_{\ell, i}$ for $\ell=0,1,\dots, s$  and $i=0,1,\dots, k-1$, we first find $Q$ in time $\poly(r^e, m)$, which determines $A_0, A_1,\dots, A_s$. Then we compute the coefficients $\hat{a}_{\ell, i}=a_{\ell, i+u}$ of the Laurent series of $A_0, A_1,\dots, A_s$  in time $\poly(r^e, m)$. 

Finally, noting that  a basis for the kernel of $M$, as well as the desired affine shift (which is any valid solution to \eqref{eq_ag_eqns}), can be found in time $\poly(r^e, m)$,  concludes the proof of Lemma \ref{lem_ag}.

\paragraph{Remarks.} We conclude this section with some remarks:

 \begin{rem}\label{rem_generic}
 For ease of presentation, we only present the construction from the Garcia--Stichtenoth tower,
 but this framework is generic and also works for other families of function fields, e.g., the Hermitian tower considered in \cite{She93, GX12}.
 Besides bounds for the genus and the number of evaluation points, we need the function fields to be explicit in the sense that there should be efficient algorithms 
for the following subroutines: computing a basis of the Riemann--Roch space $L(D)$ used in the code, evaluating a function $f\in L(D)$ at any evaluation point, and computing the Laurent series expansion of $f\in L(D)$ at a fixed rational place $P$ in a uniformizer that is fixed by the Frobenius automorphism.

In particular, if we replace $K_e$ by the rational function field $\F_q(X)$, choose the divisor $D$ to be $(k-1)P_\infty$ where $P_\infty$ denotes the unique pole of $X$, and choose the rational place $P$ for  Laurent series expansions to be the unique zero of $X$, then we recover the Reed--Solomon codes with subfield evaluation points that have been discussed in Section~\ref{sec:RS}.
 \end{rem}

\begin{rem}
We have defined two $\F_{q^m}$-linear maps, the encoding map $\mathsf{Enc}: L((k-1)P_\infty^{(m)}) \to \F_{q^m}^n$
and the embedding $\phi:  L((k-1)P_\infty^{(m)}) \to \F_{q^m}^k$ that outputs the first $k$ coefficients of the Laurent series expansion at $P_\infty^{(m)}$ in $x_e^{-1}$.
 See Figure~\ref{fig_twomaps}. Both of these two maps are efficiently computable.
 
 \begin{figure}[htb]
\centering
\begin{tikzcd}
  L((k-1)P_\infty^{(m)}) \arrow[rd,"\phi"'] \arrow[r, "\mathsf{Enc}"] &  \F_{q^m}^n \\
& \F_{q^m}^k  
\end{tikzcd}
\caption{The linear maps $\mathsf{Enc}$ and $\phi$.}
\label{fig_twomaps}
\end{figure}
 
As explained in the proof of Theorem \ref{thm:main},  the final code is defined to be $\mathsf{Enc}(\phi^{-1}(W))$ for some BTT evasive subspace  $W\subseteq  \F_{q^m}^k$. That is, we restrict the message space to $\phi^{-1}(W)$.
 
 We note that \cite{GX13} used a different idea: In \cite{GX13}, the map $\phi$ was defined on the Riemann-Roch space  $L(k'P_\infty^{(m)})$ with $k'=k-1+2g\geq k-1$. This choice of larger $k'$ guarantees that the map $\phi: L(k'P_\infty^{(m)})\to \F_{q^m}^k$ is surjective (instead of being injective).
 Then \cite{GX13} chose a subspace $V\subseteq L(k'P_\infty^{(m)})$ such that the restriction of $\phi$ to $V$ is an isomorphism between $V$ and $\F_{q^m}^k$. In this way, $\F_{q^m}^k$ may be identified with the message space $V$.  This space $V$ was further replaced by an evasive subspace  in \cite{GX13} to reduce the list size.
 
This way of restricting the message space in \cite{GX13} may be used to replace ours.
Nevertheless, we feel that our method is somewhat simpler. In particular, we only need Riemann's inequality $\ell(D)\geq \deg(D)-g+1$ in the analysis, while   \cite{GX13} uses the fact that $\ell(D)=\deg(D)-g+1$ when $\deg(D)\geq 2g-1$, which is derived from the full Riemann--Roch theorem.
\end{rem}

\paragraph{Acknowledgement.} We thank Venkatesan Guruswami for bringing our attention to \cite{GX20}.

\bibliographystyle{alpha} 
\bibliography{ref}

\newcommand{\etalchar}[1]{$^{#1}$}
\begin{thebibliography}{MRR{\etalchar{+}}20}

\bibitem[ALM{\etalchar{+}}98]{ALMSS98}
Sanjeev Arora, Carsten Lund, Rajeev Motwani, Madhu Sudan, and Mario Szegedy.
\newblock Proof verification and intractability of approximation problems.
\newblock {\em Journal of the ACM}, 45(3):501--555, 1998.

\bibitem[BFNW93]{BFNW93}
L{\'a}szl{\'o} Babai, Lance Fortnow, Noam Nisan, and Avi Wigderson.
\newblock {BPP} has subexponential time simulations unless {EXPTIME} has
  publishable proofs.
\newblock {\em Computational Complexity}, 3(4):307--318, 1993.

\bibitem[BKR10]{BKR10}
Eli {Ben-Sasson}, Swastik Kopparty, and Jaikumar Radhakrishnan.
\newblock Subspace polynomials and limits to list decoding of {Reed-Solomon}
  codes.
\newblock {\em IEEE Transactions on Information Theory}, 56(1):113--120, 2010.

\bibitem[BW87]{BW}
E.~R. Berlekamp and L.~Welch.
\newblock Error correction of algebraic block codes.
\newblock US Patent Number 4,633,470, 1987.

\bibitem[CPS99]{CPS99}
{Jin-Yi} Cai, Aduri Pavan, and D.~Sivakumar.
\newblock On the hardness of permanent.
\newblock In {\em Proceedings of the 16th Annual Symposium on Theoretical
  Aspects of Computer Science (STACS)}, volume 1563 of {\em Lecture Notes in
  Computer Science}, pages 90--99. Springer, 1999.

\bibitem[DKSS13]{DKSS13}
Zeev Dvir, Swastik Kopparty, Shubhangi Saraf, and Madhu Sudan.
\newblock Extensions to the method of multiplicities, with applications to
  {Kakeya} sets and mergers.
\newblock {\em SIAM Journal on Computing}, 42(6):2305--2328, 2013.

\bibitem[DL12]{DL12}
Zeev Dvir and Shachar Lovett.
\newblock Subspace evasive sets.
\newblock In {\em Proceedings of the 44th Annual ACM Symposium on Theory of
  Computing (STOC)}, pages 351--358. ACM Press, 2012.

\bibitem[{For}66]{For66}
David {Forney}.
\newblock {\em Concatenated Codes}.
\newblock M.I.T. Press, Cambridge, MA, USA, 1966.

\bibitem[GI01]{GI01}
Venkatesan Guruswami and Piotr Indyk.
\newblock Expander-based constructions of efficiently decodable codes.
\newblock In {\em Proceedings of the 42nd Annual IEEE Symposium on Foundations
  of Computer Science ({FOCS})}, pages 658--667. IEEE Computer Society, 2001.

\bibitem[GK16]{GK16}
Venkatesan Guruswami and Swastik Kopparty.
\newblock Explicit subspace designs.
\newblock {\em Combinatorica}, 36(2):161--185, 2016.

\bibitem[GL89]{GL89}
Oded Goldreich and Leonid~A Levin.
\newblock A hard-core predicate for all one-way functions.
\newblock In {\em Proceedings of the 21st Annual ACM Symposium on Theory of
  Computing (STOC)}, pages 25--32. ACM, 1989.

\bibitem[GQST20]{GQST20}
Fernando {Granha Jeronimo}, Dylan Quintana, Shashank Srivastava, and Madhur
  Tulsiani.
\newblock Unique decoding of explicit $\epsilon$-balanced codes near the
  gilbert-varshamov bound.
\newblock In {\em Proceedings of the 61st Annual IEEE Symposium on Foundations
  of Computer Science (FOCS)}. IEEE Computer Society, 2020.

\bibitem[GR08]{GR08_folded_RS}
Venkatesan Guruswami and Atri Rudra.
\newblock Explicit codes achieving list decoding capacity: Error-correction
  with optimal redundancy.
\newblock {\em IEEE Transactions on Information Theory}, 54(1):135--150, 2008.

\bibitem[GRS00]{GRS00}
Oded Goldreich, Dana Ron, and Madhu Sudan.
\newblock Chinese remaindering with errors.
\newblock {\em IEEE Transactions on Information Theory}, 46(4):1330--1338,
  2000.

\bibitem[GRX18]{GRX18}
Venkatesan Guruswami, Nicolas Resch, and Chaoping Xing.
\newblock Lossless dimension expanders via linearized polynomials and subspace
  designs.
\newblock In {\em Proceedings of the 33rd Computational Complexity Conference
  (CCC)}, volume 102 of {\em LIPIcs}, pages 4:1--4:16. Schloss Dagstuhl -
  Leibniz-Zentrum f{\"u}r Informatik, 2018.

\bibitem[GS96]{GS96}
Arnaldo Garcia and Henning Stichtenoth.
\newblock On the asymptotic behaviour of some towers of function fields over
  finite fields.
\newblock {\em Journal of Number Theory}, 61(2):248--273, 1996.

\bibitem[GS99]{GS-list-dec}
Venkatesan Guruswami and Madhu Sudan.
\newblock Improved decoding of {Reed-Solomon} and algebraic-geometry codes.
\newblock {\em IEEE Transactions on Information Theory}, 45(6):1757--1767,
  1999.

\bibitem[Gur09]{Gur09}
Venkatesan Guruswami.
\newblock Artin automorphisms, cyclotomic function fields, and folded
  list-decodable codes.
\newblock In {\em Proceedings of the 41st Annual ACM Symposium on Theory of
  Computing (STOC)}, pages 23--32. ACM Press, 2009.

\bibitem[GUV09]{GUV09}
Venkatesan Guruswami, Christopher Umans, and Salil Vadhan.
\newblock Unbalanced expanders and randomness extractors from {Parvaresh-Vardy}
  codes.
\newblock {\em Journal of the ACM}, 56(4):20:1--20:34, 2009.

\bibitem[GW13]{GW13}
Venkatesan Guruswami and Carol Wang.
\newblock Linear-algebraic list decoding for variants of {Reed-Solomon} codes.
\newblock {\em IEEE Transactions on Information Theory}, 59(6):3257--3268,
  2013.

\bibitem[GX12]{GX12}
Venkatesan Guruswami and Chaoping Xing.
\newblock Folded codes from function field towers and improved optimal rate
  list decoding.
\newblock In {\em Proceedings of the 44th Annual ACM Symposium on Theory of
  Computing (STOC)}, pages 339--350. ACM, 2012.

\bibitem[GX13]{GX13}
Venkatesan Guruswami and Chaoping Xing.
\newblock List decoding {Reed-Solomon}, {Algebraic-Geometric}, and {Gabidulin}
  subcodes up to the {Singleton} bound.
\newblock In {\em Proceedings of the 45th Annual ACM Symposium on Theory of
  Computing (STOC)}, pages 843--852. ACM Press, 2013.

\bibitem[GX14]{GX14}
Venkatesan Guruswami and Chaoping Xing.
\newblock Optimal rate list decoding of folded algebraic-geometric codes over
  constant-sized alphabets.
\newblock In {\em Proceedings of the 25th Annual ACM-SIAM Symposium on Discrete
  Algorithms (SODA)}, pages 1858--1866. SIAM, 2014.

\bibitem[GX15]{GX15}
Venkatesan Guruswami and Chaoping Xing.
\newblock Optimal rate algebraic list decoding using narrow ray class fields.
\newblock {\em Journal of Combinatorial Theory, Series A}, 129:160--183, 2015.

\bibitem[GX20]{GX20}
Venkatesan Guruswami and Chaoping Xing.
\newblock Optimal rate list decoding over bounded alphabets using
  algebraic-geometric codes.
\newblock {\em Electronic Colloquium on Computational Complexity (ECCC)},
  27:172, 2020.

\bibitem[HRW20]{HRW}
Brett Hemenway, Noga {Ron-Zewi}, and Mary Wootters.
\newblock Local list recovery of high-rate tensor codes and applications.
\newblock {\em SIAM Journal on Computing}, 49(4):157--195, 2020.

\bibitem[JLJ{\etalchar{+}}89]{JLJHH89}
J{\o}rn Justesen, Knud~J. Larsen, {Helge Elbr{\o}nd} Jensen, Allan Havemose,
  and Tom H{\o}holdt.
\newblock Construction and decoding of a class of algebraic geometry codes.
\newblock {\em IEEE Transactions on Information Theory}, 35(4):811--821, 1989.

\bibitem[KM93]{KM93}
Eyal Kushilevitz and Yishay Mansour.
\newblock Learning decision trees using the {Fourier} spectrum.
\newblock {\em SIAM Journal on Computing}, 22(6):1331--1348, 1993.

\bibitem[KRR{\etalchar{+}}21]{KRRSS}
Swastik Kopparty, Nicolas Resch, Noga {Ron-Zewi}, Shubhangi Saraf, and Shashwat
  Silas.
\newblock On list recovery of high-rate tensor codes.
\newblock {\em IEEE Transactions on Information Theory}, 67(1):296--316, 2021.

\bibitem[KRSW18]{KRSW}
Swastik Kopparty, Noga {Ron-Zewi}, Shubhangi Saraf, and Mary Wootters.
\newblock Improved list decoding of folded {Reed-Solomon} and multiplicity
  codes.
\newblock In {\em Proceedings of the 59th Annual IEEE Symposium on Foundations
  of Computer Science (FOCS)}, pages 212--223. IEEE Computer Society, 2018.

\bibitem[Lan02]{Lan02}
Serge Lang.
\newblock {\em Algebra}.
\newblock Springer, 2002.

\bibitem[MRR{\etalchar{+}}20]{MRRSW}
Jonathan Mosheiff, Nicolas Resch, Noga {Ron-Zewi}, Shashwat Silas, and Mary
  Wootters.
\newblock {LDPC} codes achieve list-decoding capacity.
\newblock In {\em Proceedings of the 61st Annual IEEE Symposium on Foundations
  of Computer Science (FOCS)}. IEEE Computer Society, 2020.

\bibitem[Pet60]{Peterson60}
W.~Wesley Peterson.
\newblock Encoding and error-correction procedures for the {Bose-Chaudhuri}
  codes.
\newblock {\em IRE Transactions on Information Theory}, 6(4):459--470, 1960.

\bibitem[RS60]{RS60}
Irving~S. Reed and Gustave Solomon.
\newblock Polynomial codes over certain finite fields.
\newblock {\em SIAM Journal of the Society for Industrial and Applied
  Mathematics}, 8(2):300--304, 1960.

\bibitem[RWZ20]{RWZ}
Noga {Ron-Zewi}, Mary Wootters, and Gilles Z{\'e}mor.
\newblock Linear-time erasure list-decoding of expander codes.
\newblock In {\em Proceedings of the IEEE International Symposium on
  Information Theory (ISIT)}. IEEE, 2020.

\bibitem[SAK{\etalchar{+}}01]{SAKSD01}
Kenneth~W Shum, Ilia Aleshnikov, P~Vijay Kumar, Henning Stichtenoth, and Vinay
  Deolalikar.
\newblock A low-complexity algorithm for the construction of
  algebraic-geometric codes better than the {Gilbert-Varshamov} bound.
\newblock {\em IEEE Transactions on Information Theory}, 47(6):2225--2241,
  2001.

\bibitem[She93]{She93}
B.~Z. Shen.
\newblock A {Justesen} construction of binary concatenated codes that
  asymptotically meet the {Zyablov} bound for low rate.
\newblock {\em IEEE Transactions on Information Theory}, 39(1):239--242, 1993.

\bibitem[Sti09]{Sti09}
Henning Stichtenoth.
\newblock {\em Algebraic function fields and codes}, volume 254.
\newblock Springer Science \& Business Media, 2009.

\bibitem[STV01]{STV01}
Madhu Sudan, Luca Trevisan, and Salil Vadhan.
\newblock Pseudorandom generators without the {XOR} lemma.
\newblock {\em Journal of Computer and System Sciences}, 62(2):236--266, 2001.

\bibitem[Sud97]{Sudan97}
Madhu Sudan.
\newblock Decoding of {Reed Solomon} codes beyond the error-correction bound.
\newblock {\em Journal of Complexity}, 13(1):180--193, 1997.

\bibitem[{Ta-}17]{Tashma17}
Amnon {Ta-Shma}.
\newblock Explicit, almost optimal, epsilon-balanced codes.
\newblock In {\em Proceedings of the 49th Annual ACM Symposium on Theory of
  Computing (STOC)}, pages 238--251. ACM Press, 2017.

\bibitem[Tre03]{Tre03}
Luca Trevisan.
\newblock List-decoding using the {XOR} lemma.
\newblock In {\em Proceedings of the 44th Annual IEEE Symposium on Foundations
  of Computer Science (FOCS)}, pages 126--135. IEEE Computer Society, 2003.

\bibitem[TU12]{TU12}
Amnon {Ta-Shma} and Christopher Umans.
\newblock Better condensers and new extractors from {Parvaresh-Vardy} codes.
\newblock In {\em Proceedings of the 27th Computational Complexity Conference
  (CCC)}, pages 309--315. IEEE Computer Society, 2012.

\bibitem[TZ04]{TZ04}
Amnon {Ta-Shma} and David Zuckerman.
\newblock Extractor codes.
\newblock {\em IEEE Transactions on Information Theory}, 50(12):3015--3025,
  2004.

\end{thebibliography}

 \appendix
 
\section{The Guruswami--Kopparty explicit subspace design}\label{sec:subspace-design}

In this section, for completeness, we review the proof of  Theorem \ref{thm:subspace-design-explicit}, restated below, that gives an explicit construction of a subspace design.

\sdexp*

First, we recall the definition of a subspace design.

\begin{defi}[subspace design]
An \emph{$(r, s)$-subspace design over $\F_{q^m}$ of cardinality $k$}
  is a collection of $k$ $\F_q$-linear subspaces $H_1, H_2, \dots, H_{k} \subseteq \F_{q^m}$ so that 
$\sum_{i=1}^{k} \dim (\hat{V}\cap {H_i})\leq s$
for any $\F_q$-linear subspace $\hat{V} \subseteq \F_{q^m}$ of dimension at most $r$. 
\end{defi}

Next, we sketch the construction of subspace designs in \cite{GK16}.
Let $r,t,m,q,d\in\N^+$ be such that $q$ is a prime power and $r\leq t\leq  m<q$. Let $\gamma$ be a generator of the multiplicative group $\F_q^\times$.
 For $\alpha\in\F_{q^d}$, define
 \[
 S_\alpha=\{\alpha^{q^j}\gamma^i: 0\leq j<d, 0\leq i< t\}.
 \]
\begin{lem}\label{lem_admissible}
There exists a set $\mathcal{F}\subseteq \F_{q^d}$ of cardinality at least $ \frac{q^d-1}{4dt}$ that satisfies the following conditions:
\begin{enumerate}
\item $\F_q(\alpha)=\F_{q^d}$ for $\alpha\in \mathcal{F}$.
\item $S_\alpha\cap S_\beta=\emptyset$ for distinct $\alpha,\beta\in\mathcal{F}$.
\item $|S_\alpha|=dt$ for $\alpha\in \mathcal{F}$.
\end{enumerate}
Moreover, $\mathcal{F}$ can be computed in time polynomial in $q^d$.
\end{lem}

Let  $V=\{f(X)\in \F_q[X]: \deg(f)<m\}\cong \F_{q^m}$.
For $\alpha\in\F_{q^d}$, define 
\[
H_\alpha:=\{P(X)\in V: P(\alpha\cdot \gamma^i)=0 \text{ for } j=0,1,\dots,t-1\}
\] 
which is a subspace of $V$.
As shown in \cite{GK16}, Theorem  \ref{thm:subspace-design-explicit} follows as a consequence of the following theorem. 

\begin{thm}[\cite{GK16}]\label{thm_gkgeneral}
Let $\mathcal{F}$ be as in Lemma~\ref{lem_admissible}. Then the collection $(H_\alpha)_{\alpha\in\mathcal{F}}$  is an $\left(r, s\right)$-subspace design in $V\cong\F_{q^m}$ for 
$s = \frac{(m-1)r}{d(t-r+1)}$, such that every subspace $H_\alpha$ has co-dimension at most $dt$.
\end{thm}

Note  that Theorem~\ref{thm_gkgeneral} requires the field size $q$ to be greater than $m$ while Theorem  \ref{thm:subspace-design-explicit}  does not, so the latter does not directly follow from the former.
The idea in \cite{GK16} is first using Theorem~\ref{thm_gkgeneral} to construct a subspace design in  $\F_Q^{m'}$ over an extension field $\F_Q$,  where $m'=m/ [\F_Q: \F_q]$ and $Q>m'$. (Assume $m$ is a multiple of $[\F_Q: \F_q]$ for simplicity.)
Then \cite{GK16} showed that,  by identifying $\F_Q^{m'}$ with  $\F_q^m$, this also yields a subspace design in  $\F_q^m$ with somewhat worse parameters, thereby proving Theorem  \ref{thm:subspace-design-explicit}. We refer the reader to \cite{GK16} for details.



\subsection{Proof of Lemma~\ref{lem_admissible}}

In \cite{GK16}, the set $\mathcal{F}\subseteq\F_{q^d}$ is chosen in the following way: For simplicity, assume $d$ is a prime. For $\alpha,\beta\in\F_{q^d}^\times$, write $\alpha\sim \beta$ if $\beta= \alpha^{q^i}\cdot \delta $ for some  $0\leq i<d$ and $\delta \in \F_q^\times$. Then $\sim$ is an equivalence relation on $\F_{q^d}^\times$. For each equivalence class $O\subseteq \F_{q^d}^\times$, choose a representative $\alpha_0\in O$. For $\alpha\in O$, add $\alpha$ to $\mathcal{F}$ if and only if $\alpha=\alpha_0 \gamma^{it}$ for some integer $i$ satisfying $0\leq i<\lfloor (q-1)/t\rfloor$.

However, we note that this construction of $\mathcal{F}$ does not always satisfy the conditions in Lemma~\ref{lem_admissible} when $d>1$.
For example, suppose $d$ is a prime and $q-1$ is divisible by $d$, so that $\F_q^\times$ contains all the $d$th roots of unity. In this case, $\F_{q^d}$ is a \emph{Kummer extension} $\F_q(\alpha)$ over $\F_q$ where $\alpha^d =u$ for some $u\in\F_q^\times \setminus (\F_q^\times)^d$. Then we have that $\alpha^{q-1}$ is a  $d$-th root of unity as $(\alpha^{q-1})^d = (\alpha^d)^{q-1} =u^{q-1} = 1$. By assumption that $\F_q^\times$ contains all $d$th roots of unity, this implies in turn that $\alpha^{q-1} \in \F_q^\times$.  

Let $\alpha_0 = \alpha^{q^i} \cdot \delta$ be the representative that we chose for the equivalence class of $\alpha$, where $0 <i <d$ and $\delta \in \F_q^\times$. Then we claim that $\alpha_0^{q-1} \in \F_q^\times$ as
$ \alpha_0^{q-1} = (\alpha^{q-1})^{q^i}\cdot \delta^{q-1} = \alpha^{q-1} \in \F_q^\times.$ Consequently, we have that $\alpha_0^q=\alpha_0 \gamma^{it+j}$ for some integers $i$ and $j$ with $0\leq i<\lceil(q-1)/t\rceil$ and $0\leq j<t$. If $0<i<\lfloor (q-1)/t\rfloor$, then $\alpha_0$ and $\alpha_0\gamma^{it}$ are distinct and both added to $\mathcal{F}$. This violates the second condition in Lemma~\ref{lem_admissible}  since we have $\alpha_0^q\in S_{\alpha_0}$ and $\alpha_0^q=\alpha_0 \gamma^{it+j}\in S_{\alpha_0\gamma^{it}}$, which implies $S_{\alpha_0}\cap S_{\alpha_0\gamma^{it}}\neq \emptyset$.
Similarly, if $i=0$, then the third condition $|S_\alpha|=dt$ does not hold.

One way of fixing this problem is ignoring those elements $\alpha\in\F_{q^d}^\times$ satisfying $\alpha^{q^i-1}\in\F_q^\times$ for some $0<i<d$. 
The next lemma gives an upper bound for the number of those elements.

\begin{lem}\label{lem_boundB}
Let $B=\{\alpha\in\F_{q^d}^\times: \alpha^{q^i-1}\in\F_q^\times \text{ for some } 0<i<d \}$. Then $|B|\leq (q^d-1)/2$.\footnote{This bound is attained when $q=3$ and $d=2$ but not tight in general. We have made no attempt to optimize the bound.}
\end{lem}

\begin{proof}
If $d=1$, then $|B|=0\leq (q^d-1)/2$ . So assume $d\geq 2$.
Consider $\alpha\in B$. We have $\alpha^{q^i-1}=\delta$ for some $0<i<d$ and $\delta\in\F_q^\times$.
Note $\alpha^{q^{d-i}-1}=(1/\delta)^{q^{d-i}}=1/\delta$. So by replacing $(i, \delta)$ with $(d-i, 1/\delta)$ if necessary, we may assume $i\leq d/2$.

For any $\alpha'\in\F_{q^d}^\times$ satisfying $(\alpha')^{q^i-1}=\delta$, we have $(\alpha'/\alpha)^{q^i}=\alpha'/\alpha$ and hence $\alpha'/\alpha\in  \F_{q^i}^\times$.
So the number of $\alpha'\in\F_{q^d}^\times$ satisfying  $(\alpha')^{q^i-1}=\delta$ is at most $q^{i}-1$.
Therefore, for fixed $\delta\in\F_q^\times$, the number of $\alpha\in\F_{q^d}^\times$ for which there exists an integer $0<i\leq d/2$ satisfying $\alpha^{q^i-1}=\delta$ is bounded by
\[
N:=\sum_{i=1}^{\lfloor d/2\rfloor} (q^{i}-1)
=(q^{\lfloor d/2\rfloor+1}-q)/(q-1)-\lfloor d/2\rfloor.
\]
There are $q-1$ choices of $\delta\in\F_q^\times$.
So we have 
\[
|B|\leq (q-1)N = q^{\lfloor d/2\rfloor+1}-q-\lfloor d/2\rfloor(q-1).
\]
When $d\geq 3$, we have $q^{\lfloor d/2\rfloor+1}\leq q^{d-1}\leq q^{d}/2$ and hence $|B|\leq (q^d-1)/2$, as desired. 

Now assume $d=2$. We need a more careful analysis in this case. Note that if $\delta\in\F_q^\times$ can be written as $\alpha^{q-1}$ then $\delta^{q+1} = \alpha^{(q - 1)(q+1)} = \alpha^{q^2-1} =1$, i.e., $\delta$ is a $(q+1)$th root of unity. The number of such $\delta\in\F_q^\times$ equals $\gcd(q+1, q-1)=\gcd(q+1,2)$. So we have 
\begin{equation}\label{eq_r2}
|B|\leq \gcd(q+1,2)N=\gcd(q+1,2)\cdot (q-1).
\end{equation}
It is easy to see that the RHS of \eqref{eq_r2} is at most $(q^2-1)/2$.
So $|B|\leq (q^2-1)/2=(q^d-1)/2$. 
\end{proof}

We now give a complete proof of  Lemma~\ref{lem_admissible}.

\begin{proof}[Proof of  Lemma~\ref{lem_admissible}] 
Let $B$ be as in Lemma~\ref{lem_boundB}.
For $\alpha,\beta\in\F_{q^d}^\times$, write $\alpha\sim \beta$ if $\beta\in \alpha^{q^i}\F_q^\times$ for some $0\leq i<d$. Then $\sim$ is an equivalence relation on $\F_{q^d}^\times$. Note that if $\alpha\sim \beta$ and $\alpha\in  B$, then $\beta\in B$. So $\F_{q^d}^\times \setminus B$ is a disjoint union of equivalence classes under the relation $\sim$. Moreover, the definition of $B$ implies that for every $\alpha\in \F_{q^d}^\times \setminus B$, the equivalence class $O_\alpha=\{\alpha^{q^j} \gamma^i: 0\leq j<d, 0\leq i<q-1\}$ of $\alpha$ has cardinality exactly $d(q-1)$. So the number of equivalence classes contained in $\F_{q^d}^\times \setminus B$ is $\frac{q^d-1-|B|}{d(q-1)}\geq \frac{q^d-1}{2d(q-1)}$, where we use the bound $|B|\leq (q^d-1)/2$ given by Lemma~\ref{lem_boundB}.

Construct $\mathcal{F}$ as follows: For each equivalence class $O\subseteq \F_{q^d}^\times \setminus B$, fix a representative $\alpha_0\in O$, and add $\alpha_0 \gamma^{it}$ to $\mathcal{F}$ for $i=0,1,\dots, \lfloor (q-1)/t\rfloor-1$. 
We have
\[
|\mathcal{F}|\geq   \frac{q^d-1}{2d(q-1)} \cdot  \left\lfloor \frac{q-1}{t}\right\rfloor \geq \frac{q^d-1}{2d(q-1)} \cdot   \frac{q-1}{2t}= \frac{q^d-1}{4dt}.
\]
Clearly, $\mathcal{F}$ can be computed in time polynomial in $|\F_{q^d}|=q^d$.  

We have $S_\alpha\cap S_\beta=\emptyset$ for distinct $\alpha,\beta\in\mathcal{F}$ and $|S_\alpha|=dt$ for $\alpha\in \mathcal{F}$. This follows from the fact that every equivalence class $O=\{\alpha_0^{q^j} \gamma^i: 0\leq j<d, 0\leq i<q-1\}\subseteq \F_{q^d}^\times \setminus B$ has cardinality exactly $d(q-1)$, i.e., the $d(q-1)$ elements $\alpha_0^{q^j} \gamma^i$ are distinct.

Finally, we also have $\F_q(\alpha)=\F_{q^d}$ for $\alpha\in\mathcal{F}$. This follows from the fact that $\alpha^{q^i-1}\neq 1$ for $\alpha\in \F_{q^d}^\times \setminus B$ and $0<i<d$, which holds by the definition of $B$.
\end{proof}

\end{document}